\documentclass{article}
\usepackage[margin=1.3in]{geometry}
\usepackage{cite}
\usepackage{amssymb}
\usepackage{amsmath}
\usepackage{amsthm}
\usepackage[T1]{fontenc}
\usepackage{latexsym}
\usepackage{mathrsfs}
\usepackage{tikz}
\usepackage{hyperref}
\usepgflibrary{arrows}
\usetikzlibrary{matrix,arrows,calc}
\usepackage{xcolor}

\usepackage[caption=false,font=normalsize,labelfont=sf,textfont=sf]{subfig}
\usepackage{subfiles}
\usepackage[capitalize]{cleveref}

\usepackage{booktabs}

\newtheorem{thm}{Theorem}
\newtheorem{prop}{Proposition}

\newtheorem{dfn}{Definition}
\newtheorem{cor}{Corollary}

\theoremstyle{definition}

\newcommand{\dd}{\,\textrm{d}}
\newcommand{\R}{\mathbb{R}}

\renewcommand{\L}{\mathcal{L}}
\newcommand{\F}{\mathcal{F}}

\newcommand{\autocite}[1]{#1}

\definecolor{ggplot1}{RGB}{27,158,119}
\definecolor{ggplot2}{RGB}{217,95,2}
\definecolor{ggplot3}{RGB}{117,112,179}
\definecolor{ggplot4}{RGB}{231,41,138} 
\title{Local Independence Testing for Point Processes}

\author{Nikolaj Thams, Niels Richard Hansen
\thanks{The work presented in this article is supported in part by Novo Nordisk Foundation Grant NNF20OC0062897 and in part by VILLUM FONDEN Grant 18968. \textit{(Corresponding author: Nikolaj Thams.)}}
\thanks{N. Thams and N. R. Hansen are with the Department of Mathematics, University of Copenhagen, Denmark (e-mail: thams@math.ku.dk; niels.r.hansen@math.ku.dk)}
}

\begin{document}
\maketitle
\begin{abstract} Constraint based causal structure learning for point processes require empirical tests of local independence. Existing tests require strong model assumptions, e.g. that the true data generating model is a Hawkes process with no latent confounders.
Even when restricting attention to Hawkes processes, latent confounders are a major technical difficulty because a marginalized process will generally not be a Hawkes process itself. We introduce an expansion similar to Volterra expansions as a tool to represent marginalized intensities. Our main theoretical result is that such expansions can approximate the true marginalized intensity arbitrarily well. Based on this we propose a test of local independence and investigate its properties in real and simulated data. 
\end{abstract}

\section{Introduction}
Hawkes
processes are models of time-dynamic interacting point processes with applications in such diverse areas as finance \cite{Bacry:2015}, seismology \cite{Ogata:2013}, social science \cite{Zhou:2013} and neuroscience \cite{Truccolo:2005}. Hawkes proposed himself that his model for self- and mutually exciting point processes could be applied as a model of epidemic spread and neuron firing among other things \cite{Hawkes:1971}, and with reference to Hawkes' pivotal work the model has taken the name \emph{Hawkes process} in the literature. Specifically, Hawkes introduced the multivariate linear Hawkes process, which together with its nonlinear extension \cite{Bremaud:1996} have become the most widely applied models of multivariate dynamic point processes.

It is straightforward to define -- in purely mathematical terms -- whether one event type in a Hawkes process affects another event type. This defines a network, and our main objective is to test hypotheses regarding network connectivity. Constraining the network structure to be sparse can have well known statistical and computational benefits, e.g. a favourable bias-variance tradeoff for large networks and fast data fitting algorithms \cite{Hansen:2015b,Chen:2017}. However, it is much less obvious if the network structure allows for a subject matter interpretation beyond the purely statistical one. In particular, if the network conveys causal information.

We will use Hawkes process models of neuron spike activity as a main motivating example, and we will discuss the question of causal discovery in this context, though our results are of a general nature. Hawkes processes have a long history in neuron science with Brillinger using them some 45 years ago for the first time \cite{Brillinger:1975, Brillinger:1976}. Early applications relied on moment identities and spectral methods, but likelihood methods later became computationally feasible and widely used \cite{Brillinger:1992, Brown:2004aa, Truccolo:2005, Pillow:2008,Hansen:2015a}. The Hawkes processes have served several objectives, from a statistical characterization of dependencies among correlated neurons to a vehicle for sensory decoding from neuron ensembles, and, more recently, as a way to learn a sparse network structure among the neurons \cite{Masud:2011, Song:2013}.
According to \cite{Song:2013} the Hawkes process can identify the \emph{functional connectivity} of a neural network, but the network ``cannot be directly interpreted as synaptic connections'' -- yet the model's attractiveness was from the very beginning tied to its physiological interpretability as representing synaptic integration \cite{Brillinger:1976}. Moreover, functional connectivity was interpreted in \cite{Song:2013} as a \emph{causal relation}, and understanding the extent to which this interpretation is justified was a main motivation for the work presented in this paper.

The notion of a ``causal relation'' was left undefined by \cite{Song:2013}, and it is possible that it was only meant in the  weak sense of Granger causality as considered earlier by e.g. \cite{Kim:2011} for neuron spike activity. Irrespectively, it is of interest to understand if stronger causal interpretations are possible, e.g. identification of intervention effects as expressed by \cite{Didelez:2008, Didelez:2015} and \cite{Roysland:2012}. The methods proposed by \cite{Song:2013}, as well as most methods in the statistical literature \cite{Hansen:2015b, Hall:2016, Eichler:2017,   Chen:2017, Achab:2017}, result in networks that only allow for a strong causal interpretation by assuming that all variables of the system are observed. Causal structure learning algorithms like Meek's Causal Analysis (CA) algorithm \cite{Meek:2014} likewise require all variables observed, but recent constraint based learning algorithms  \cite{Mogensen:2018uai, Mogensen:2020} do allow for a strong causal interpretation of the resulting network even in the presence of latent confounders. The algorithm by \cite{Mogensen:2018uai} is related to FCI for acyclic causal structures \cite{Spirtes:1993}, but it is adapted to cyclic graphs that can represent time-dynamic feedback mechanisms. 
Where FCI and other algorithms for acyclic graphs are relying on tests of conditional independence, cyclic graphs of time-dynamic systems are based  on (conditional) \emph{local independence} \cite{Didelez:2008,Mogensen:2020}, and causal discovery algorithms require empirical tests of this asymmetric independence criterion.

In this paper, we propose a test of local independence in point process data. Let $j$ and $k$ denote two types of events, e.g. the firing of two different neurons, and let $C$ denote a set of event types, e.g. a set of neurons. The hypothesis that $k$ is locally independent of $j$ given $C$ is denoted $j \not\rightarrow k \mid C$. We test the hypothesis by testing whether events of type $j$ contribute significantly to the intensity of $k$ given events of type $C$. 
We approximate point process intensities by basis expansions and propose to use higher-order interactions terms of events to fit intensities, such that the intensity does not only take into account single events (as is the case for Hawkes processes), but also pairs or triples of events. We show that 
higher-order interactions can be captured through iterated integrals and that any intensity can be arbitrarily well approximated by including enough higher-order terms, analogous to Volterra expansions in dynamical systems \cite{volterra1959theory,franz2006unifying}.

Our main motivation for this nonparametric expansion is that Hawkes processes are not closed under marginalization, meaning that a subcollection of event types of a Hawkes process need not be a Hawkes process. 
Consequently, if some event types of a Hawkes process are unobserved, we may not be able to model it by a Hawkes process (that is, by using only first-order terms of events). 
Even if all processes are observed, constraint based learning algorithms \cite{Meek:2014,Mogensen:2018uai} construct a local independence graph by testing $j \not\rightarrow k \mid C$ within a (typically small) subcollection of event types, in effect corresponding to testing local independence when marginalizing away everything else than $j$, $k$ and $C$.
The model misspecification arising from the assumption that the marginalized processes are Hawkes may result in tests that do not have asymptotic level. By including higher-order interactions in our tests, this model misspecification is reduced, such that the null hypotheses of local independence are rejected less often, resulting in sparser and more correct graphs. 

\subsection{Structure of this article}
In \cref{sec:background} we outline the existing theory on Hawkes processes and local independence. \Cref{sec:volterra} contains our main theoretical result, that intensities can be approximated arbitrarily well by including higher-order interaction terms. We apply this approximation in \cref{sec:test} to construct a test of local independence. 
In \cref{sec:simulation} we evaluate the test in simulation studies, and in \cref{sec:real-world} we apply the test in causal learning algorithms to learn network structure in a neuron spiking data set.

\section{Hawkes processes and local independence}\label{sec:background}
In this section we first give a brief introduction to Hawkes processes (see \cite{Daley:2003} for a more thorough introduction). We then introduce local independence graphs and tests of local independence.

\begin{figure*}[t]
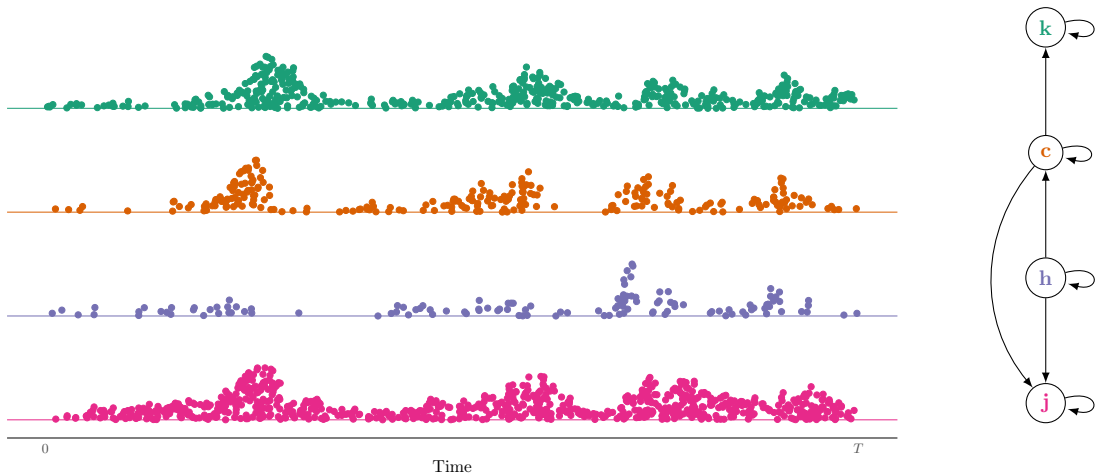

    \subfloat{
	\centering
        \resizebox{0.8\textwidth}{!}{
%
		\label{intro:fig_data_set_b}%
    }
	\caption{(\textit{left}): Data from a $4$-dimensional Hawkes process. The vertical position of points reflect the local frequency of points. (\textit{right}): Local independence graph (see \cref{subsec:loc-indep-graph}) of the process that generated the data. }
	\label{intro:fig_data_set}
\end{figure*}

Let $V = \{1, \ldots, d\}$ and let $N = (N^k)_{k \in V}$ denote a collection of point processes on $\mathbb{R}$ indexed by $V$. Each mark $k\in V$ represents a particular type of event, and $N$ is also referred to as a marked point process \cite{Daley:2003}.
More formally, if for each $k = 1, \ldots, d$, we let $\{\ldots, \tau_{-1}^k, \tau_{0}^k, \tau_1^k, \tau_2^k,\ldots\}$ be a series of event times, the point process $N^k$ is defined as the random measure
\begin{equation*}
    N^k(A) = \sum_{i} \delta_{\tau^{k}_i}(A),
\end{equation*}
where $\delta_t(A)$ is a Dirac measure.
We associate with the \mbox{$k$-th} point process the counting process $N_t^k = N^k((0,t])$. We will assume that $N$ is simple and non-exploding, meaning that event times are distinct and any finite interval only has finitely many points. 

For a point process $N$, the intensity $\lambda_t = (\lambda_t^1, \ldots, \lambda^d_t)$ describes the conditional rate of new events at time $t$,
\begin{align*}
    \lambda^k_t = \mathbb{E}(N^k(\dd t) \mid \mathcal{F}_{t-}^V),
\end{align*}
where $\F^V_{t-}$ is the predictable filtration generated by $N^1, \ldots, N^d$, i.e. $\F^V_{t-}$ is the history of events strictly prior to time $t$ of any type $j \in V$.

Let $g^{jk} : [0,\infty) \rightarrow [0,\infty)$ for $j,k \in V$ be integrable
functions, which we call kernels. We introduce the intensity process 
\begin{align} 
    \lambda_t^k &= \beta_0^k + \sum_{j \in V} \sum_{i: \tau^j_i < t} g^{jk}(t-\tau^j_i) \nonumber \\
    &=\beta_0^k + \sum_{j \in V} \int_{-\infty}^{t-} g^{jk}(t-s) N^{j}(\mathrm{d} s)\label{eq:Hawkes_int}
\end{align}
where we call $\beta_0^k \geq 0$ baseline intensities.

\begin{dfn} A $d$-dimensional point process with intensity processes
  $\lambda^k$, $k \in V$, as defined by \eqref{eq:Hawkes_int} is called a multivariate linear Hawkes process with kernels $g^{jk}$ and
  baseline intensities $\beta_0^k$.
\end{dfn}

In this paper, we only consider stationary Hawkes processes. If we define
\begin{equation}\label{eq:integrated-kernel}
    \mathbf{g}^{jk} = \int_{0}^{\infty} g^{jk}(t) \mathrm{d} t,
\end{equation}
and define the matrix $G = (\mathbf{g}^{jk})_{j,k}$, stationarity for the linear Hawkes process is achieved if the spectral radius of $G$ is strictly smaller than $1$, see
\cite{Hawkes:1974} and Chapter 6 in \cite{Daley:2003}. 

The linear Hawkes process can be extended to \emph{the nonlinear Hawkes process} using a link function $\eta$:
\begin{equation*}
    \eta(\lambda_t^k) = \beta_0^k + \sum_{j \in V} \int_{-\infty}^{t-} g^{jk}(t-s) N^{j}(\mathrm{d} s).
\end{equation*}
Useful alternatives to $\eta(x) = x$ are $\eta(x)=\log(x)$ or \mbox{$\eta(x) = 1_{x\geq 1}\cdot x + 1_{x<1}\cdot (\log(x)+1)$.} In both cases, $\eta^{-1}$ maps $\R$ into $[0, \infty)$, which ensures that $\lambda^k_t \geq 0$ even if the kernels are allowed to take negative values. In the following subsection, we only discuss marginalization in the linear Hawkes process. However, the approximation result in \cref{sec:volterra} extends readily to nonlinear processes, and so we state that result in generality.

\subsection{Marginalization in Hawkes processes}\label{subsec:marginalization}
If we only observe events corresponding to marks in $V' \subset V$,
the distribution of the $V'$-events is a marginalization of the 
distribution of $V$-events. Even if all event types of a system are observed, the local independence statement $j \not\rightarrow k \mid C$ relates to the marginal distribution of $N^j, N^k$ and $N^C$, so when $\{j,k\}\cup C \neq V$, we test local independence in a marginalized distribution.

This creates a problem for testing, because many model classes, including Hawkes processes, are not closed under marginalization, i.e. the marginalized distribution need not be in the same model class as before. 
We explore the case of marginalized Hawkes processes in more detail. 

For $C \subseteq V$ let $\mathcal{F}^C_{t-} := \cup_{s<t}\F_s^C$ denote the predictable filtration generated by $N^j$ for $j \in C$, and let $E( \cdot \mid \mathcal{F}^C_{t-})$ denote expectations given only information about events strictly prior to $t$ of types $C$.\footnote{Technically, $E( \cdot \mid \mathcal{F}^C_{t-})$ is the predictable projection operator in order to have regular sample paths of the resulting stochastic process. See the remark in \cite{florens1996noncausality} for a discussion of this.} Suppose that $k \in C$, then by the innovation theorem (see e.g. \cite{jacobsen2006point}) the $\mathcal{F}^C_{t-}$-intensity of $N^k$ is  
\begin{equation*}
    \lambda_t^{k,C} = E( \lambda_t^{k} \mid \mathcal{F}^C_{t-}).
\end{equation*}
We will refer to this as the $C$-intensity. 
For the linear Hawkes process we have from \eqref{eq:Hawkes_int} that
\begin{align} 
    \lambda_t^{k,C} &= \beta_0^k + \sum_{j \in C} \int_{-\infty}^{t-}
    g^{jk}(t-s) N^{j}(\mathrm{d} s) \label{eq:Hawkes_filt}\\
    &\quad + \sum_{l \in C^c} \int_{-\infty}^{t-} g^{lk}(t-s) E(N^{l} \mid \mathcal{F}^C_{t-}) (\mathrm{d} s), \nonumber
\end{align}
thus for a complete computation of the $C$-intensity we need to
compute $E(N^{l} \mid \mathcal{F}^C_{t-})$, which is a classical
filtering problem. The solution can be characterized via general
filtration equations, see \cite{arjas1992filtering} and \cite{last1995marked}.

We could approximate the solution of the filtering problem by a linear filter
\begin{align*}
    &E(N^{l} \mid \mathcal{F}^C_{t-})(\mathrm{d}s) \simeq \\ 
    & \quad \bigg(\gamma^l_0 \quad + \sum_{j \in C} \int_{-\infty}^{t-} h^{jl}(t - u, t - s) N^j(\mathrm{d}u) \bigg) \mathrm{d}s,
\end{align*}
for a choice of kernels $h^{jl}$. Using this, we arrive at the following \emph{approximate} $C$-intensity 
\begin{equation} \label{eq:Hawkes_filt_approx}
    \tilde{\lambda}_t^{k,C} = \tilde{\beta}_0^{k,C} + \sum_{j \in C} \int_{-\infty}^{t-}\tilde{g}^{jk, C}(t-s) N^{j}(\mathrm{d} s) 
\end{equation}
where 
$$\tilde{\beta}_0^{k, C} = \beta_0^k + \sum_{l \in C^c} \int_0^{\infty}
\gamma_0^l g^{lk}(t) \mathrm{d} t$$
and
$$\tilde{g}^{jk, C}(t) = g^{jk}(t) + \sum_{l \in C^c} \int_{0+}^{\infty}
h^{jl}(t, s) g^{lk}(s) \mathrm{d}s$$ 
for $j \in C$. We recognise \eqref{eq:Hawkes_filt_approx} as being the intensity for a linear Hawkes process over event types indexed by $C$.
However, this is only an approximation, and the marginalized process will generally not be a linear Hawkes process. Thus some effects of this \emph{model misspecification} should be expected if we fit a model of the form
\eqref{eq:Hawkes_filt_approx} to marginalized data. 

\subsection{Local independence hypotheses}
Following \cite{Mogensen:2020}, we define local independence for a point process $N$ by saying that $N^k$ is locally independent of $N^j$ given $N^C$ if $\lambda^{k, C\cup\{j\}}$ has an $\mathcal{F}^{C}_{t}$-predictable version. Intuitively that means that $\lambda^{k, C\cup\{j\}}$ only depends on events in $N^C$ and not $N^j$. In this case, we write $j \not \rightarrow k \mid C$, and else (if $\lambda_t^{k,C}$ is not a version of $\lambda_t^{k,C\cup \{j\}}$) we write $j \rightarrow k \mid C$. 

Our goal is to test the local independence hypothesis, 
\begin{align*}
    H_0: j \not \rightarrow k \mid C.
\end{align*}
In the approximate $C$-intensity from \cref{eq:Hawkes_filt_approx} this hypothesis corresponds  to $\tilde{g}^{jk, C \cup \{j\}}$ being 0. 
However, a test of $\tilde{g}^{jk, C \cup \{j\}} = 0$ as a surrogate for $H_0$ comes with no guarantee on the level due to the model misspecification of $\tilde{\lambda}^{k,C}$.

Instead of relying on the first-order approximation in \cref{eq:Hawkes_filt_approx} for $\lambda^{k,C}$, we propose to base 
the test on the approximation 
\begin{equation}\label{eq:approximate_intensity}
    \overline{\lambda}^{k, C \cup \{j\}}_t = \lambda^{k, C}_t + \int_{-\infty}^{t-}
    \overline{g}^{jk}(t - s) N^j(\mathrm{d} s)
\end{equation}
of the $ C \cup \{j\}$-intensity $\lambda^{k, C\cup\{j\}}$. This 
approximation only uses a linear filter to model the contribution from $j$
under the alternative, and under $H_0$ the model with $\overline{g}^{jk} = 0$ is, in fact, correctly specified.
Thus we will carry out tests of $H_0$ by testing $\overline{g}^{jk} = 0$. A major practical and technical challenge is to approximate and fit $\lambda^{k, C}$ sufficiently well for the test to maintain level, and we dedicate \cref{sec:volterra} to developing methods for appropriately fitting $\lambda^{k, C}$. 

\subsection{Local independence graphs}\label{subsec:loc-indep-graph}
In this paper, we consider tests for local independence, with the motivation of learning graphical representations of causal relations in point processes.
In particular, we consider the local independence graph for point processes, introduced by Didelez in \cite{Didelez:2008}, where the absence of an edge $j\not\rightarrow k$ in the graph corresponds to the local independence $j\not\rightarrow k \mid V\setminus \{j,k\}$.

For the linear Hawkes process, the local independence graph is a graph with vertices $V$ and an edge $j \rightarrow k$ if and only if $\mathbf{g}^{jk} > 0$, where $\mathbf{g}^{jk}$ is defined in \cref{eq:integrated-kernel}. That is, there is an edge from $j$ to $k$ if and only if the kernel $g^{jk}$ is not constantly equal to 0. \Cref{intro:fig_data_set} displays data from a Hawkes process and the underlying local independence graph that was used to generate the data.

\section{Higher-order expansions}\label{sec:volterra}
\subsection{Motivating higher-order interactions}
In the following, we propose a general expansion of point process intensities, which we show to converge to the true intensity as the degree of the expansion approaches infinity. We intend to apply this to marginalized Hawkes processes, in order to remove the model misspecification discussed above, but the result does not rely on the process being Hawkes, and applies to any point process model.

The expansion utilizes iterated integrals, which already \cite{Brillinger:1975} used for specifying models with higher-order interactions. \cite{Cohen:2012} showed that the chaos expansion of point processes initiated at zero can approximate any measurable variable arbitrarily well, by integrals over random intervals. Similar to \cite{Cohen:2012} our proof relies on martingale convergence, but uses integrals over deterministic intervals. 

Iterated integrals are also used in the theory of Volterra series \cite{volterra1959theory}, where the dynamics of a time-homogeneous system over variables $x$ and $y$ is approximated by the $\mathit{L^{\textrm{th}}}$\textit{-order} expansion: 
\begin{align*}
y_t &\approx \beta^0 + \\&
\sum_{n=1}^L \int_{-\infty}^t \! \! \cdots \int_{-\infty}^t h^n(s_1, \ldots, s_n) x_{t-s_1} \cdots x_{t-s_n} \dd s_1 \cdots \dd s_n
\end{align*}
Under various regularity conditions, including continuity and finite memory of the system, this approximation will converge, that is, the right hand side converges to $y_t$ for all $t$ when $L$ tends to infinity \cite{ahmed1970closure,franz2006unifying}. Although point process systems are very different in nature to continuous systems, we show a similar expansion for point processes below. 

\subsection{Intensity representations}
We consider a fixed subset $C \subseteq V$, and a stationary process $N$. 
Let $C_n$ be the set of tuples $\alpha = (j_1, \ldots, j_n)$ of length $n$ where $j_i \in C$ and $j_{i_1} \leq j_{i_2}$ for $i_1 < i_2$. 
Further define
\begin{equation*}
    E_n = \mathcal{L}([0, \infty)^{n}, \mathbb{R})^{C_n},
\end{equation*}
where $\L([0,\infty)^n, \R)$ is the set of measurable functions $h: [0,\infty)^n \to \R$. 
That is, every element $(h^\alpha)_{\alpha\in C_n}$ in $E_n$ is a collection of functions, indexed by the distinct combinations of $C$. 
Also define the functional $\varphi_n^{t}$ on $E_n$:
\begin{equation*}
	 \varphi_n^{t}: (h^\alpha)_{\alpha \in C_n} \mapsto \sum_{\alpha \in C_n} \int_{(-\infty, t)} h^\alpha (t-s^n) N^\alpha(\dd s^n)
\end{equation*}
where 
\begin{align*}
&\int_{(-\infty, t)} g(t-s^n) N^\alpha(\mathrm{d}s^n) := \\ &\quad \int_{-\infty}^{t-} \cdots \int_{-\infty}^{t-} g(t - s_1, \ldots, t-s_n) N^{j_1}(\dd s_1) \cdots N^{j_n}(\dd s_n).
\end{align*}
Note that $\varphi_n^{t}$ maps into $\mathcal{L}(\mathcal{F}_{t-}^C)$ because for any $h:= (h^{\alpha})_{\alpha \in C_n}$, the filter $\varphi_n^{t}(h)$ is $\mathcal{F}_{t-}^C$-measurable.

For example, if $V = \{1, 2\}$ and $n = 2$, we have $C_n = \{(1,1), (1,2), (2,2)\}$, $E_n = \{(h^{(1,1)},h^{(1,2)}, h^{(2,2)})\mid h^{(i,j)}: [0, \infty)^2 \to \R \textrm{ measurable}\}$ and 
\begin{align*}
    &\varphi_n^{t}(h^{(1,1)},h^{(1,2)}, h^{(2,2)}) \\
    & = \int_{-\infty}^{t-}\int_{-\infty}^{t-} h^{(1,1)}(t-s_1, t-s_2) N^1(\dd s_1)N^1(\dd s_2) \\
    & + \int_{-\infty}^{t-}\int_{-\infty}^{t-} h^{(1,2)}(t-s_1, t-s_2) N^1(\dd s_1)N^2(\dd s_2) \\
    & + \int_{-\infty}^{t-}\int_{-\infty}^{t-} h^{(2,2)}(t-s_1, t-s_2) N^2(\dd s_1)N^2(\dd s_2),
\end{align*}
is the evaluation of the kernels $h^{(1,1)}, h^{(1,2)}$ and $h^{(2,2)}$ in all combinations of points in the respective event types $N^1$ and $N^2$.

We now show that we can approximate point process intensities by such sums of iterated integrals. 
We first show this for $t=0$ and then extend the result to all $t \in \R$ using time homogeneity. 
At $t=0$, we define the set $W_n$ of all $\mathcal{F}_{t-}^C$-measurable random variables, that can be written as a $n$-fold iterated integral and are almost surely finite: 
\begin{align*}
W_n = \{X \in \varphi_n^{0}(E_n) \mid |X| < \infty \text{ a.s} \}
\end{align*}

This allows us to state the following theorem, which is proven in the appendix. 
\begin{thm}\label{thm:representation_0}
With $\mathcal{F}_{0-}^C = \sigma(\cup_{s<0}\F_s^C)$ it holds that 
$\bigoplus_{n\in\mathbb{N}} W_n$ is dense in $\{X \in \mathcal{L}(\mathcal{F}_{0-}^C) \mid |X| < \infty \text{ a.s.}\}$ in the topology of convergence in probability. \footnote{i.e. the topology induced by the Ky Fan metric $$d(X, Y) = \inf\{\epsilon > 0 \mid P(|X - Y| > \epsilon) \leq \epsilon\}.$$}
\end{thm}

That is, every finite $\mathcal{F}_{0-}^C$-measurable variable can be approximated arbitrarily well by iterated integrals, over the past events of the processes in $C$.

Consider now the case of a point process intensity $\lambda^{k,C}_t$, and let $\eta$ link function.
Assume further that the intensity is time homogeneous: if $\eta(\lambda_t^{k,C})(\tau_1, \tau_2, \ldots)$ denotes the mechanism with which $\eta(\lambda^{k,C}_t)$ depends on the event times prior to time $t$, we say that $\eta(\lambda^{k,C}_t)$ is time-homogeneous if for $s \geq 0$, 
\begin{equation*}
    \eta(\lambda^{k,C}_t)(\tau_1, \tau_2, \ldots) = \eta(\lambda^{k,C}_{t-s})(\tau_1 -s, \tau_2-s, \ldots).
\end{equation*}

\begin{cor}\label{thm:volterra_rep}
	If $\eta(\lambda^{k,C}_t)$ is a time homogeneous point process intensity, $\eta(\lambda^{k,C}_t)$ can at all times be arbitrarily well approximated by iterated integrals in the topology of convergence in probability. 
\end{cor}
\begin{proof}
	Take $\epsilon > 0$ and any $t \in \mathbb{R}$. Since the intensity is $\F^C_t$-predictable, $\eta(\lambda_0^{k,C}) \in \F_{0-}^C$ at time $t=0$.
	Thus take $\phi^0 \in \bigoplus_{n\in \mathbb{N}} W_n$ such that $P(|\eta(\lambda^{k,C}_0) - \phi^0| > \epsilon) < \epsilon$, which is possible by \cref{thm:representation_0}. 
	Since $\bigoplus_{n\in \mathbb{N}} W_n$ is a sum of images, we can choose $h_1 \in E_1, h_2 \in E_2, \ldots, h_L \in E_L$ such that $\phi^0 = \sum_{n=1}^L \varphi_n^{0}(h_n)$. Let $\phi$ be the process $t \mapsto \sum_{n=1}^L \varphi_n^{t}(h_n)$, and observe that $\phi$ is time homogeneous. 
	
	Conclusively, the process $\eta(\lambda^{k,C}_t) - \phi^t$ is time homogeneous, and by the assumed stationarity, the distribution of $\eta(\lambda^{k,C}_t) - \phi^t$ is invariant over $t$. In particular $P(|\eta(\lambda^{k,C}_t) - \phi^t| > \epsilon) < \epsilon$ for all $t \in \mathbb{R}$. 
\end{proof}
Observe that it is the same kernels $h_1, \ldots, h_L$ that 
enter into the approximation of $\lambda^{k,C}_t$ for all $t$. In \cref{thm:representation_0}, there is nothing special about $t=0$, and one could as well have proven that $\bigoplus_{n\in\mathbb{N}} \varphi_n^{t}(E_n)$ is dense in $\mathcal{L}(\mathcal{F}_{t-}^C)$. However, only by the time-homogeneity can one be ensured that the same kernels can be used for all $t$. 

\subsection{Approximate intensities}
The fully observed (nonlinear) Hawkes process has intensity given by sums of first-order terms 
\begin{equation*}
\eta(\lambda^k_t) = \eta(\lambda^{k,V}_t) = \beta^0 + \sum_{j\in V} \int_{-\infty}^{t-} g^{jk}(t-s) N^j(\dd s).
\end{equation*}
As discussed in \cref{subsec:marginalization}, when $C \neq V$, $\lambda^{k, C}$ cannot in general be represented by sums of first-order terms. However, by \cref{thm:volterra_rep} the intensity can be approximated by including interaction terms of higher orders, and so one could approximate $\lambda^{k,C}$ by the $\mathit{L^{\textrm{th}}}$\textit{-order} expansion
\begin{align*}
    \eta(\lambda_t^{k,C}) \approx \beta_0^k + \sum_{n=1}^L \sum_{\alpha \in C_n} \int_{(-\infty,t)} h_n^{\alpha}(t-s^n) N^{\alpha}(\dd s^n)
\end{align*}
for some sequence of kernels $h_n^\alpha$, $1\leq n \leq L$, $\alpha \in C_n$. For $L = 2$ we obtain the approximate intensity: 
\begin{align}\label{eq:second_order}
	&\eta(\lambda_t^{k, C}) \approx 
	\beta_0^k \\
	&+ \sum_{j_1 \in C}\int_{-\infty}^{t-} h^{j_1}(t-s_1) N^{j_1}(\dd s_1) \notag \\ 
    &+ \sum_{j_1, j_2 \in C}\int_{-\infty}^{t-} \int_{-\infty}^{t-} h^{j_1, j_2}(t-s_1, t-s_2) N^{j_1}(\dd s_1) N^{j_2}(\dd s_2) \notag
\end{align}
The class of models described by \eqref{eq:second_order} contains the class of linear Hawkes processes (corresponding to $h^{j_1, j_2} = 0$) but also encompasses more complicated models, such as a model where the intensity boosts only when two events occur very close to each other.

\section{Testing local independence}\label{sec:test}
We now return to the question of developing a test for local independence $j \not\rightarrow k \mid C$. We consider the approximation of $\lambda^{k,C\cup\{j\}}$ in \cref{eq:approximate_intensity}, and use the higher-order interactions from \cref{sec:volterra} together with basis splines to approximate $\lambda^{k,C}$. We fit this approximation from data and test significance of the contribution from $j$. 

\subsection{Approximating kernel functions}
We consider the question of approximating the intensities $\int_0^{t-} \overline{g}^{jk} N^j(\dd s)$ and $\lambda_t^{k, C}$ from \eqref{eq:approximate_intensity}. 

To approximate the intensity $\lambda_t^{k, C}$, we utilize the $W_0 \oplus W_1 \oplus W_2$-approximation from \cref{eq:second_order}. We approximate the kernels $h^{j_1}(s_1)$ and $h^{j_1, j_2}(s_1, s_2)$ by spline expansions 
\begin{equation*}
h^{j_1} \approx \sum_{i} \beta_i^{j_1} b_i \quad \textrm{and}\quad  h^{j_1, j_2} \approx \sum_{i_1, i_2} \beta_{i_1, i_2}^{j_1, j_2} b_{i_1} \otimes b_{i_2},
\end{equation*}
for some class of basis functions $\{b_i\}_{i}$ such as B-splines \cite{friedman2001elements}. 
Due to the linearity in $\beta$, the coefficient terms can be collected into one vector $\beta^C$ and we can write $\lambda_t^{k, C}\approx \left(\beta^C\right)^T x_t^C$. 
Each entry of $x_t^C$ corresponds to one basis function integrated with respect to either a single event type or a pair of event types. For instance the entry corresponding to $\beta_{i_1, i_2}^{j_1, j_2}$ would be 
\begin{equation*}
    \int_{-\infty}^{t-} \int_{-\infty}^{t-} b_{i_1}(t- s_1) b_{i_2}(t-s_2) N^{j_1}(\dd s_1) N^{j_2}(\dd s_2).
\end{equation*}

Similarly, we approximate the kernel $\overline{g}^{jk}$ by $\sum_i \overline{\beta}^j_i b_i$, and collect the coefficients to $\overline{\beta}^j$ and $\overline{x}_t^j$. Conclusively, the intensity \eqref{eq:approximate_intensity} can be approximated by
\begin{equation*}
    \eta\left(\overline{\lambda}_t^{k, C\cup\{j\}}\right) = \left(\beta^C\right)^T x_t^C + \left(\overline{\beta}^j\right)^T \overline{x}_t^j =: \left(\beta^{C \cup\{j\}}\right)^T x_t^{C\cup \{j\}}
\end{equation*}
for some choice of $\beta^C$ and $\overline{\beta}^j$. 
 
\subsection{Maximum likelihood}
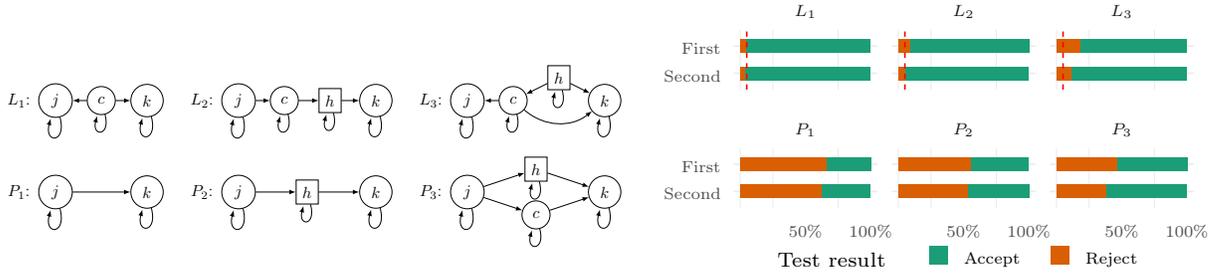
\begin{figure*}[ht]
        \subfloat{
		\centering
		\resizebox{0.55\textwidth}{!}{%
		\begin{tikzpicture}[>=latex,font=\sffamily,every node/.style={inner sep = 4pt}]
			\node[draw, circle] (A) at (4, 0) {$j$};
		    \node[draw, circle] (C) at (5,0) {$c$};
		    \node[draw, circle] (B) at (6,0) {$k$};
		    \path [->] (A) edge [loop below] (A);
		    \path [->] (B) edge [loop below] (B);
		    \path [->] (C) edge [loop below] (C);
		    \path [->] (C) edge (A);
		    \path [->] (C) edge (B);

			\node[draw, circle] (A) at (8, 0) {$j$};
		    \node[draw, circle] (C1) at (9,0) {$c$};
		    \node[draw] (C2) at (10,0) {$h$};
		    \node[draw, circle] (B) at (11,0) {$k$};
		    \path [->] (A) edge [loop below] (A);
		    \path [->] (B) edge [loop below] (B);
		    \path [->] (C1) edge [loop below] (C1);
		    \path [->] (C2) edge [loop below] (C2);
		    \path [->] (A) edge (C1);
		    \path [->] (C1) edge (C2);
		    \path [->] (C2) edge (B);
		    \node[draw, circle] (A) at (13, 0) {$j$};
		    \node[draw, circle] (C1) at (14,0) {$c$};
		    \node[draw] (C2) at (15,0.5) {$h$};
		    \node[draw, circle] (B) at (16,0) {$k$};
		    \path [->] (A) edge [loop below] (A);
		    \path [->] (B) edge [loop below] (B);
		    \path [->] (C1) edge [loop below] (C1);
		    \path [->] (C2) edge [loop below] (C2);
		    \path [->] (C1) edge (A);
		    \path [->] (C2) edge (C1);
		    \path [->] (C2) edge (B);
		    \path [->] (C1) edge[bend right=40] (B);
		    \node at (3.2, 0) {$L_1$:};
		    \node at (7.2, 0) {$L_2$:};
		    \node at (12.2, 0) {$L_3$:};

			\node[draw, circle] (A) at (4, -2) {$j$};
		    \node[draw, circle] (B) at (6, -2) {$k$};
		    \path [->] (A) edge [loop below] (A);
		    \path [->] (B) edge [loop below] (B);
		    \path [->] (A) edge (B);
		    \node[draw, circle] (A) at (8, -2) {$j$};
		    \node[draw] (C2) at (9.5,-2) {$h$};
		    \node[draw, circle] (B) at (11, -2) {$k$};
		    \path [->] (A) edge [loop below] (A);
		    \path [->] (B) edge [loop below] (B);
		    \path [->] (C2) edge [loop below] (C2);
		    \path [->] (A) edge (C2);
		    \path [->] (C2) edge (B);

			\node[draw, circle] (A) at (13, -2) {$j$};
		    \node[draw, circle] (C1) at (14.5, -2.5) {$c$};
		    \node[draw] (C2) at (14.5, -1.5) {$h$};
		    \node[draw, circle] (B) at (16, -2) {$k$};
		    \path [->] (A) edge [loop below] (A);
		    \path [->] (B) edge [loop below] (B);
		    \path [->] (C1) edge [loop below] (C1);
		    \path [->] (C2) edge [loop below] (C2);
		    \path [->] (A) edge (C1);
		    \path [->] (A) edge (C2);
		    \path [->] (C1) edge (B);
		    \path [->] (C2) edge (B);

		    \node at (3.2, -2) {$P_1$:};
		    \node at (7.2, -2) {$P_2$:};
		    \node at (12.2, -2) {$P_3$:};
		    \node at (5, -3.5){};%
		\end{tikzpicture}}%
		}%
		\hspace{0.03\linewidth}%
		\subfloat{
		    \centering
		  \scriptsize
\begin{tikzpicture}[x=1pt,y=1pt]
\definecolor{fillColor}{RGB}{255,255,255}
\begin{scope}
\definecolor{drawColor}{gray}{0.92}

\path[draw=drawColor,line width= 0.3pt,line join=round] ( 39.80, 63.67) --
	( 39.80, 86.33);

\path[draw=drawColor,line width= 0.3pt,line join=round] ( 52.15, 63.67) --
	( 52.15, 86.33);

\path[draw=drawColor,line width= 0.3pt,line join=round] ( 76.84, 63.67) --
	( 76.84, 86.33);

\path[draw=drawColor,line width= 0.6pt,line join=round] ( 37.33, 69.85) --
	( 91.66, 69.85);

\path[draw=drawColor,line width= 0.6pt,line join=round] ( 37.33, 80.15) --
	( 91.66, 80.15);
\definecolor{fillColor}{RGB}{217,95,2}

\path[fill=fillColor] ( 39.80, 77.58) rectangle ( 42.27, 82.73);

\path[fill=fillColor] ( 39.80, 67.28) rectangle ( 41.88, 72.43);
\definecolor{fillColor}{RGB}{27,158,119}

\path[fill=fillColor] ( 42.27, 77.58) rectangle ( 89.19, 82.73);

\path[fill=fillColor] ( 41.88, 67.28) rectangle ( 89.19, 72.43);
\definecolor{drawColor}{RGB}{255,0,0}

\path[draw=drawColor,line width= 0.6pt,dash pattern=on 2pt off 2pt ,line join=round] ( 42.27, 63.67) -- ( 42.27, 86.33);
\end{scope}
\begin{scope}
\definecolor{drawColor}{gray}{0.92}

\path[draw=drawColor,line width= 0.3pt,line join=round] ( 39.80, 18.94) --
	( 39.80, 41.60);

\path[draw=drawColor,line width= 0.3pt,line join=round] ( 52.15, 18.94) --
	( 52.15, 41.60);

\path[draw=drawColor,line width= 0.3pt,line join=round] ( 76.84, 18.94) --
	( 76.84, 41.60);

\path[draw=drawColor,line width= 0.6pt,line join=round] ( 37.33, 25.12) --
	( 91.66, 25.12);

\path[draw=drawColor,line width= 0.6pt,line join=round] ( 37.33, 35.42) --
	( 91.66, 35.42);
\definecolor{fillColor}{RGB}{217,95,2}

\path[fill=fillColor] ( 39.80, 32.84) rectangle ( 72.50, 37.99);

\path[fill=fillColor] ( 39.80, 22.54) rectangle ( 70.82, 27.69);
\definecolor{fillColor}{RGB}{27,158,119}

\path[fill=fillColor] ( 72.50, 32.84) rectangle ( 89.19, 37.99);

\path[fill=fillColor] ( 70.82, 22.54) rectangle ( 89.19, 27.69);
\end{scope}
\begin{scope}
\definecolor{drawColor}{gray}{0.92}

\path[draw=drawColor,line width= 0.3pt,line join=round] ( 99.63, 63.67) --
	( 99.63, 86.33);

\path[draw=drawColor,line width= 0.3pt,line join=round] (111.97, 63.67) --
	(111.97, 86.33);

\path[draw=drawColor,line width= 0.3pt,line join=round] (136.67, 63.67) --
	(136.67, 86.33);

\path[draw=drawColor,line width= 0.6pt,line join=round] ( 97.16, 69.85) --
	(151.48, 69.85);

\path[draw=drawColor,line width= 0.6pt,line join=round] ( 97.16, 80.15) --
	(151.48, 80.15);
\definecolor{fillColor}{RGB}{217,95,2}

\path[fill=fillColor] ( 99.63, 77.58) rectangle (104.07, 82.73);

\path[fill=fillColor] ( 99.63, 67.28) rectangle (102.39, 72.43);
\definecolor{fillColor}{RGB}{27,158,119}

\path[fill=fillColor] (104.07, 77.58) rectangle (149.01, 82.73);

\path[fill=fillColor] (102.39, 67.28) rectangle (149.01, 72.43);
\definecolor{drawColor}{RGB}{255,0,0}

\path[draw=drawColor,line width= 0.6pt,dash pattern=on 2pt off 2pt ,line join=round] (102.10, 63.67) -- (102.10, 86.33);
\end{scope}
\begin{scope}
\definecolor{drawColor}{gray}{0.92}

\path[draw=drawColor,line width= 0.3pt,line join=round] ( 99.63, 18.94) --
	( 99.63, 41.60);

\path[draw=drawColor,line width= 0.3pt,line join=round] (111.97, 18.94) --
	(111.97, 41.60);

\path[draw=drawColor,line width= 0.3pt,line join=round] (136.67, 18.94) --
	(136.67, 41.60);

\path[draw=drawColor,line width= 0.6pt,line join=round] ( 97.16, 25.12) --
	(151.48, 25.12);

\path[draw=drawColor,line width= 0.6pt,line join=round] ( 97.16, 35.42) --
	(151.48, 35.42);
\definecolor{fillColor}{RGB}{217,95,2}

\path[fill=fillColor] ( 99.63, 32.84) rectangle (127.28, 37.99);

\path[fill=fillColor] ( 99.63, 22.54) rectangle (126.00, 27.69);
\definecolor{fillColor}{RGB}{27,158,119}

\path[fill=fillColor] (127.28, 32.84) rectangle (149.01, 37.99);

\path[fill=fillColor] (126.00, 22.54) rectangle (149.01, 27.69);
\end{scope}
\begin{scope}
\definecolor{drawColor}{gray}{0.92}

\path[draw=drawColor,line width= 0.3pt,line join=round] (159.45, 63.67) --
	(159.45, 86.33);

\path[draw=drawColor,line width= 0.3pt,line join=round] (171.80, 63.67) --
	(171.80, 86.33);

\path[draw=drawColor,line width= 0.3pt,line join=round] (196.49, 63.67) --
	(196.49, 86.33);

\path[draw=drawColor,line width= 0.6pt,line join=round] (156.98, 69.85) --
	(211.31, 69.85);

\path[draw=drawColor,line width= 0.6pt,line join=round] (156.98, 80.15) --
	(211.31, 80.15);
\definecolor{fillColor}{RGB}{217,95,2}

\path[fill=fillColor] (159.45, 77.58) rectangle (168.74, 82.73);

\path[fill=fillColor] (159.45, 67.28) rectangle (165.28, 72.43);
\definecolor{fillColor}{RGB}{27,158,119}

\path[fill=fillColor] (168.74, 77.58) rectangle (208.84, 82.73);

\path[fill=fillColor] (165.28, 67.28) rectangle (208.84, 72.43);
\definecolor{drawColor}{RGB}{255,0,0}

\path[draw=drawColor,line width= 0.6pt,dash pattern=on 2pt off 2pt ,line join=round] (161.92, 63.67) -- (161.92, 86.33);
\end{scope}
\begin{scope}
\definecolor{drawColor}{gray}{0.92}

\path[draw=drawColor,line width= 0.3pt,line join=round] (159.45, 18.94) --
	(159.45, 41.60);

\path[draw=drawColor,line width= 0.3pt,line join=round] (171.80, 18.94) --
	(171.80, 41.60);

\path[draw=drawColor,line width= 0.3pt,line join=round] (196.49, 18.94) --
	(196.49, 41.60);

\path[draw=drawColor,line width= 0.6pt,line join=round] (156.98, 25.12) --
	(211.31, 25.12);

\path[draw=drawColor,line width= 0.6pt,line join=round] (156.98, 35.42) --
	(211.31, 35.42);
\definecolor{fillColor}{RGB}{217,95,2}

\path[fill=fillColor] (159.45, 32.84) rectangle (182.47, 37.99);

\path[fill=fillColor] (159.45, 22.54) rectangle (178.52, 27.69);
\definecolor{fillColor}{RGB}{27,158,119}

\path[fill=fillColor] (182.47, 32.84) rectangle (208.84, 37.99);

\path[fill=fillColor] (178.52, 22.54) rectangle (208.84, 27.69);
\end{scope}
\begin{scope}
\definecolor{drawColor}{gray}{0.10}

\node[text=drawColor,anchor=base,inner sep=0pt, outer sep=0pt, scale=  0.88] at ( 64.50, 46.85) {$P_1$};
\end{scope}
\begin{scope}
\definecolor{drawColor}{gray}{0.10}

\node[text=drawColor,anchor=base,inner sep=0pt, outer sep=0pt, scale=  0.88] at (124.32, 46.85) {$P_2$};
\end{scope}
\begin{scope}
\definecolor{drawColor}{gray}{0.10}

\node[text=drawColor,anchor=base,inner sep=0pt, outer sep=0pt, scale=  0.88] at (184.15, 46.85) {$P_3$};
\end{scope}
\begin{scope}
\definecolor{drawColor}{gray}{0.10}

\node[text=drawColor,anchor=base,inner sep=0pt, outer sep=0pt, scale=  0.88] at ( 64.50, 91.59) {$L_1$};
\end{scope}
\begin{scope}
\definecolor{drawColor}{gray}{0.10}

\node[text=drawColor,anchor=base,inner sep=0pt, outer sep=0pt, scale=  0.88] at (124.32, 91.59) {$L_2$};
\end{scope}
\begin{scope}
\definecolor{drawColor}{gray}{0.10}

\node[text=drawColor,anchor=base,inner sep=0pt, outer sep=0pt, scale=  0.88] at (184.15, 91.59) {$L_3$};
\end{scope}
\begin{scope}
\definecolor{drawColor}{gray}{0.30}

\node[text=drawColor,anchor=base,inner sep=0pt, outer sep=0pt, scale=  0.88] at ( 64.50,  7.93) {50\%};

\node[text=drawColor,anchor=base,inner sep=0pt, outer sep=0pt, scale=  0.88] at ( 89.19,  7.93) {100\%};
\end{scope}
\begin{scope}
\definecolor{drawColor}{gray}{0.30}

\node[text=drawColor,anchor=base,inner sep=0pt, outer sep=0pt, scale=  0.88] at (124.32,  7.93) {50\%};

\node[text=drawColor,anchor=base,inner sep=0pt, outer sep=0pt, scale=  0.88] at (149.01,  7.93) {100\%};
\end{scope}
\begin{scope}
\definecolor{drawColor}{gray}{0.30}

\node[text=drawColor,anchor=base,inner sep=0pt, outer sep=0pt, scale=  0.88] at (184.15,  7.93) {50\%};

\node[text=drawColor,anchor=base,inner sep=0pt, outer sep=0pt, scale=  0.88] at (208.84,  7.93) {100\%};
\end{scope}
\begin{scope}
\definecolor{drawColor}{gray}{0.30}

\node[text=drawColor,anchor=base east,inner sep=0pt, outer sep=0pt, scale=  0.88] at ( 32.38, 66.82) {Second};

\node[text=drawColor,anchor=base east,inner sep=0pt, outer sep=0pt, scale=  0.88] at ( 32.38, 77.12) {First};
\end{scope}
\begin{scope}
\definecolor{drawColor}{gray}{0.30}

\node[text=drawColor,anchor=base east,inner sep=0pt, outer sep=0pt, scale=  0.88] at ( 32.38, 22.09) {Second};

\node[text=drawColor,anchor=base east,inner sep=0pt, outer sep=0pt, scale=  0.88] at ( 32.38, 32.39) {First};
\end{scope}
\begin{scope}
\definecolor{drawColor}{RGB}{0,0,0}

\node[text=drawColor,anchor=base west,inner sep=0pt, outer sep=0pt, scale=  1.10] at ( 53.85, -3.43) {Test result};
\end{scope}
\begin{scope}
\definecolor{fillColor}{RGB}{27,158,119}

\path[fill=fillColor] (111.23, -2.61) rectangle (118.34,  4.50);
\end{scope}
\begin{scope}
\definecolor{fillColor}{RGB}{217,95,2}

\path[fill=fillColor] (157.16, -2.61) rectangle (164.27,  4.50);
\end{scope}
\begin{scope}
\definecolor{drawColor}{RGB}{0,0,0}

\node[text=drawColor,anchor=base west,inner sep=0pt, outer sep=0pt, scale=  0.88] at (124.55, -2.08) {Accept};
\end{scope}
\begin{scope}
\definecolor{drawColor}{RGB}{0,0,0}

\node[text=drawColor,anchor=base west,inner sep=0pt, outer sep=0pt, scale=  0.88] at (170.48, -2.08) {Reject};
\end{scope}
\end{tikzpicture}
 		}
		\caption{(\textit{left}) Graphical structures used for testing local independence. Square nodes indicate unobserved event types. For each, we simulate $500$ samples from a Hawkes process with this true local independence graph, and evaluate the test $j \not\rightarrow k \mid C$ for with $C$ being $\{c,k\}$ or (in the absence of a node $c$) $\{k\}$.
		(\textit{right}) $H_0$ acceptance rates ($p < 0.05$ level) for the 500 repetitions of the test $j \rightarrow k \mid \{k, c\}$ in each of the structures using both a first- (1) and second- (2) order approximation of $\lambda^{k, C}$. The colors indicate the proportion of tests accepted and rejected, and the dashed line marks $5\%$ rejection rate (only relevant for graphs $L_1$-$L_3$).}
		\label{fig:level_power}
\end{figure*}
Given an observation of a point process over the interval $[0, T]$, we compute maximum likelihood estimates $\hat\beta^{C \cup \{j\}}$ using the penalized log-likelihood
\begin{equation*}
    \int_0^T  \log \overline{\lambda}_t^{k, C\cup \{j\}} N^k(\dd t) - \int_0^T \overline{\lambda}_t^{k, C\cup \{j\}} \dd t - \rho(\beta^{C\cup\{j\}}),
\end{equation*}
where $\rho(\beta) = \kappa_0 \beta^T \Omega \beta$ is a quadratic penalization, and where $\kappa_0 > 0$ and $\Omega$ is the roughness penalty matrix, which penalizes curvature of the kernel estimates (see Chapter 5 in~\cite{friedman2001elements}).

Assuming that the true model belongs to the model class, with parameter $\beta_0^{C\cup \{j\}}$, it follows from \cite{Hansen:2015a} that the distribution of the maximum likelihood estimate $\hat\beta^{C\cup \{j\}}$ is approximately normal with mean 
$$\mu = (I + 2\kappa_0 \hat{J}_T^{-1}\Omega)\beta_0^{C\cup \{j\}}$$
and covariance matrix 
$$\Sigma = \hat J_T^{-1} \hat K_T \hat J_T^{-1}$$
where 
\begin{align*}
\hat K_T &= \int_0^T x_t^{C\cup \{j\}} {x_t^{C\cup \{j\}}}^T \frac{\left((\eta^{-1})'(\hat \beta x_t^{C\cup \{j\}})\right)^2}{\eta^{-1}(\hat \beta x_t^{C\cup \{j\}})} \dd t \\
\hat J_T &= \hat K_T - 2 \kappa_0 \Omega.
\end{align*}

If $\mu_j, \Sigma_j$ denotes the respective subvector and -matrix which corresponds to the entries of $\overline\beta^j$, the approximate distribution of the estimated parameter $\hat{\overline\beta}^j$ is known and can be used for testing. 

\subsection{Hypothesis testing}
We can now test the hypothesis $H_0: \overline g^{jk} = 0$ by testing whether $\hat{\overline\beta^j} = 0$. 
In the setting of testing $g = 0$ for a function $g = \sum_i \beta_i b_i$, \cite{wood2012p} show that directly testing $\hat \beta^j = 0$ can lead to loss of power. Instead,  \cite{wood2012p} proposes to evaluate the function in a grid $\mathbb{X} = (x_1, \ldots, x_M)$ and perform the hypothesis test that the resulting vector $g(\mathbb{X}):= (g(x_m))_{1 \leq m \leq M}$ is $0$. 

Let $\mathbb{B} = (b_i(x_m))_{m, i}$ be the matrix where the $i$-th column is the evaluation of the $i$-th basis function evaluated in $\mathbb{X}$. Then $\overline g^{jk}(\mathbb{X}) = \mathbb{B}\hat{\overline\beta}^{jk}$ is the evaluation of $\overline g^{jk}$ in $\mathbb{X}$, which is then approximately $\mathcal{N}\left(\mathbb{B}\mu_j, \mathbb{B} \Sigma_j \mathbb{B}^T \right)$-distributed. 
This allows for testing the hypothesis $\overline g^{jk} = 0$ by the Wald-test statistic: 
$$T = \left[\mathbb{B}\hat{\overline\beta}^{jk}\right]^T \left(\mathbb{B} \Sigma_j \mathbb{B}^T \right)^{-1} \left[\mathbb{B}\hat{\overline\beta}^{jk}\right]$$ which is approximately $\chi^2_{(M)}$-distributed. By comparing $T$ to the theoretical quantiles of $\chi^2_{(M)}$, we can test for significance of the contribution of $j$ to the intensity $\lambda^{k, C\cup\{j\}}$. 
The test is implemented in \verb|python| and is available online.\footnote{Code available at \url{https://github.com/nikolajthams/LIPP}.}

\section{Simulation experiments}\label{sec:simulation}
We evaluate our test using simulated data. First we explore the level and power for several graphical structures. Second we apply the test in a causal discovery algorithm to learn the local independence graph from an observed data set. 
In both experiments, we compare our method to the first-order method in \cref{eq:Hawkes_filt_approx}, where also the $\lambda^{k, C}$ intensity is approximated by basis expansions using only first-order interaction terms. 

\subsection{Level and power}\label{subsec:level-power}
In \cref{subsec:marginalization} we argued that the misspecification from using only first-order terms may lead to a loss of level. 
To validate this, for each of the graphs $\mathcal{G}$ in \cref{fig:level_power}, we sample $n=500$ point processes from the Hawkes process with kernel $g^{i_1i_2}(s) = \alpha_{i_1i_2}\beta_{i_1i_2}e^{-\beta_{i_1i_2} s}$ if $(i_1, i_2) \in \mathcal{G}$ and otherwise $g^{i_1i_2}(s)=0$.
Simulation details are in \cref{appendix:sim-details}.

For each sample, we test the hypothesis \mbox{$H_0: j\not\rightarrow k \mid C$} with $C = \{c, k\}$ (or $C = \{k\}$ in the graphs with no node $c$). The hypothesis $H_0$ is true in structures $L_1$--$L_3$ (and thus we here evaluate level) and false in structures $P_1$--$P_3$ (and so we here evaluate power).

The nodes $h$ represent an unobserved event type, and so is not included in the conditioning set $C$. Due to the latent events, we expect the first-order test to loose level compared to the second-order test.
We conduct the test of $H_0$ from \cref{sec:test} on a nominal $5\%$ level and display in \cref{fig:level_power} the proportion of $p$-values below $5\%$ for each structure, with red indicating a rejected test of $H_0$. 

In the structure $L_1$, we observe that the both the first- and second order tests maintain level in the structure $L_1$. This is as expected, because the ground truth structure $L_1$ has no 
latent events, and so the effect $c \rightarrow k$ is truly a first-order interaction. 
In the structure $L_2$, our proposed second order test has a rejection rate around $5\%$, while the first-order test exceeds the nominal level by rejecting in around $9\%$ of the simulations. This indicates that due to the latent process $N^h$ being marginalized out, the dependence between $N^c$ and $N^k$ is not fully captured by first-order interactions, and so when fitting only first-order interactions, there is some residual information which mistakenly is then captured in the fitted kernel $\overline{g}^{jk}$. 
By introducing second-order interactions, this residual information is reduced, and the false negative link $j\rightarrow k$ becomes less likely. 
In $L_3$ both the first- and second-order tests reject in more than $5\%$ of cases, however with the level of the second-order test being closer to the nominal $5\%$ level. This indicates that the marginalization of $h$ induces a model misspecification which is partly captured by the second-order interaction. 

For the graphs $P_1, P_2$ and $P_3$, where truly $j \rightarrow k \mid C$, we observe that both the first- and second-order approaches have substantial power. For the structure $P_3$, we observe that the first-order test has more power than the second-order test, possibly due to the fewer parameters that need to be estimated to use the first-order test.

\subsection{Causal Structure Learning}\label{subsec:sim-eca}
We also evaluate the proposed test in the context of the Causal Analysis (CA) algorithm proposed by Meek \cite{Meek:2014}, which is similar to the PC-algorithm \cite{Spirtes:1993} but applies to local independence graphs (see \cref{fig:illustration_pcalg} for an illustration of the algorithm). For $d \in \{3, \ldots, 7\}$, we simulate $n=60$ graphs of dimension $d$ and with each edge occurring with a fixed probability of $0.2$. We then simulate a Hawkes process with the simulated graph as causal graph. Simulation details are in \cref{subsec:details-eca}.
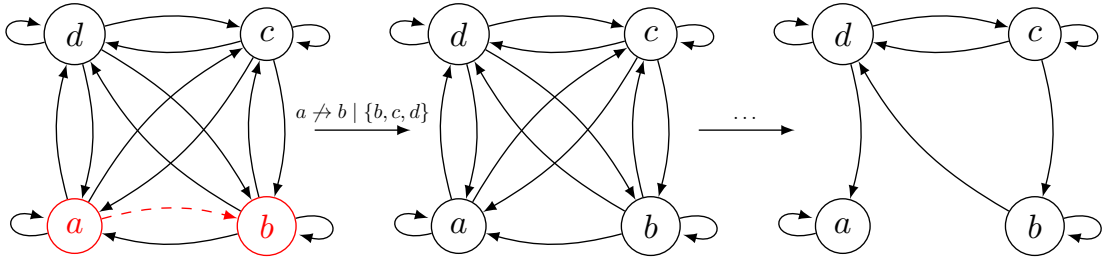
\begin{figure}[t]
		\centering
		\resizebox{\linewidth}{!}{%
		\begin{tikzpicture}[>=latex,font=\sffamily]
			\node[draw, circle,red] (A) at (0, 0) {$a$};
		    \node[draw, circle, red] (B) at (2, 0) {$b$};
		    \node[draw, circle] (C) at (2, 2) {$c$};
		    \node[draw, circle] (D) at (0, 2) {$d$};
		    \path [->, dashed, red] (A) edge [bend left = 15] (B);
		    \path [->] (B) edge [bend left = 15] (A);
		    \path [->] (A) edge [bend left = 15] (C);
		    \path [->] (C) edge [bend left = 15] (A);
		    \path [->] (A) edge [bend left = 15] (D);
		    \path [->] (D) edge [bend left = 15] (A);
		    \path [->] (B) edge [bend left = 15] (C);
		    \path [->] (C) edge [bend left = 15] (B);
		    \path [->] (B) edge [bend left = 15] (D);
		    \path [->] (D) edge [bend left = 15] (B);
		    \path [->] (C) edge [bend left = 15] (D);
		    \path [->] (D) edge [bend left = 15] (C);
		    \path [->] (A) edge [loop left] (A);
		    \path [->] (B) edge [loop right] (B);
		    \path [->] (C) edge [loop right] (C);
		    \path [->] (D) edge [loop left] (D);
		    \path [->] (2.5, 1) edge node[above, midway,scale=0.6] {$a \not\to b \mid \{b, c, d\}$} (3.5, 1);
			\node[draw, circle] (A) at (4, 0) {$a$};
		    \node[draw, circle] (B) at (6, 0) {$b$};
		    \node[draw, circle] (C) at (6, 2) {$c$};
		    \node[draw, circle] (D) at (4, 2) {$d$};
		    \path [->] (B) edge [bend left = 15] (A);
		    \path [->] (A) edge [bend left = 15] (C);
		    \path [->] (C) edge [bend left = 15] (A);
		    \path [->] (A) edge [bend left = 15] (D);
		    \path [->] (D) edge [bend left = 15] (A);
		    \path [->] (B) edge [bend left = 15] (C);
		    \path [->] (C) edge [bend left = 15] (B);
		    \path [->] (B) edge [bend left = 15] (D);
		    \path [->] (D) edge [bend left = 15] (B);
		    \path [->] (C) edge [bend left = 15] (D);
		    \path [->] (D) edge [bend left = 15] (C);
		    \path [->] (A) edge [loop left] (A);
		    \path [->] (B) edge [loop right] (B);
		    \path [->] (C) edge [loop right] (C);
		    \path [->] (D) edge [loop left] (D);
		    \path [->] (6.5, 1) edge node[above,midway,scale=0.6] {$\cdots$} (7.5, 1);
			\node[draw, circle] (A) at (8, 0) {$a$};
		    \node[draw, circle] (B) at (10, 0) {$b$};
		    \node[draw, circle] (C) at (10, 2) {$c$};
		    \node[draw, circle] (D) at (8, 2) {$d$};
		    \path [->] (D) edge [bend left = 15] (A);
		    \path [->] (C) edge [bend left = 15] (B);
		    \path [->] (B) edge [bend left = 15] (D);
		    \path [->] (C) edge [bend left = 15] (D);
		    \path [->] (D) edge [bend left = 15] (C);
		    \path [->] (A) edge [loop left] (A);
		    \path [->] (B) edge [loop right] (B);
		    \path [->] (C) edge [loop right] (C);
		    \path [->] (D) edge [loop left] (D);
		\end{tikzpicture} %
}
		\label{fig:illustration_pcalg}
		\caption{Illustration of constrained based learning algorithms like the Causal Analysis algorithm \cite{Meek:2014} or the PC-algorithm \cite{Spirtes:1993}. The algorithm starts with the fully connected graph (\textit{left}), and removes the edge $a \rightarrow b$ if there exist a set $C$ of current parents of $b$, such that $a \not\rightarrow b \mid C$ (\textit{middle}). This is then done repeatedly for all nodes and for sets $C$ of increasing size. The algorithm terminates, when no more edges can be removed, that is when no more local independences can be found (\textit{right}).}
\end{figure}
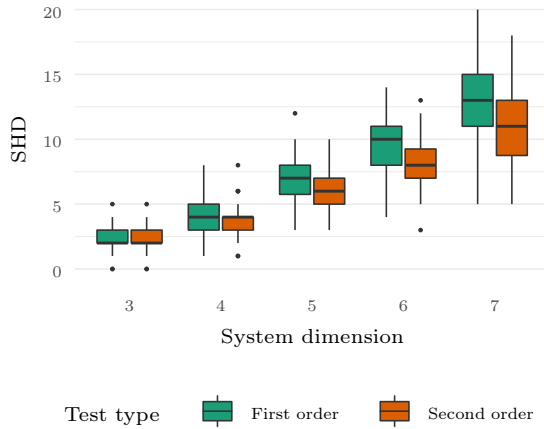
\begin{figure}[hbt]
  \center
  \scriptsize
\begin{tikzpicture}[x=1pt,y=1pt]
\definecolor{fillColor}{RGB}{255,255,255}
\begin{scope}
\definecolor{drawColor}{gray}{0.92}

\path[draw=drawColor,line width= 0.3pt,line join=round] ( 31.71, 84.33) --
	(211.31, 84.33);

\path[draw=drawColor,line width= 0.3pt,line join=round] ( 31.71,108.88) --
	(211.31,108.88);

\path[draw=drawColor,line width= 0.3pt,line join=round] ( 31.71,133.43) --
	(211.31,133.43);

\path[draw=drawColor,line width= 0.3pt,line join=round] ( 31.71,157.99) --
	(211.31,157.99);

\path[draw=drawColor,line width= 0.6pt,line join=round] ( 31.71, 72.05) --
	(211.31, 72.05);

\path[draw=drawColor,line width= 0.6pt,line join=round] ( 31.71, 96.60) --
	(211.31, 96.60);

\path[draw=drawColor,line width= 0.6pt,line join=round] ( 31.71,121.16) --
	(211.31,121.16);

\path[draw=drawColor,line width= 0.6pt,line join=round] ( 31.71,145.71) --
	(211.31,145.71);

\path[draw=drawColor,line width= 0.6pt,line join=round] ( 31.71,170.26) --
	(211.31,170.26);
\definecolor{drawColor}{gray}{0.20}
\definecolor{fillColor}{gray}{0.20}

\path[draw=drawColor,line width= 0.4pt,line join=round,line cap=round,fill=fillColor] ( 45.96, 72.05) circle (  0.68);

\path[draw=drawColor,line width= 0.4pt,line join=round,line cap=round,fill=fillColor] ( 45.96, 72.05) circle (  0.68);

\path[draw=drawColor,line width= 0.4pt,line join=round,line cap=round,fill=fillColor] ( 45.96, 96.60) circle (  0.68);

\path[draw=drawColor,line width= 0.6pt,line join=round] ( 45.96, 86.78) -- ( 45.96, 91.69);

\path[draw=drawColor,line width= 0.6pt,line join=round] ( 45.96, 81.87) -- ( 45.96, 76.96);
\definecolor{fillColor}{RGB}{27,158,119}

\path[draw=drawColor,line width= 0.6pt,line join=round,line cap=round,fill=fillColor] ( 40.13, 86.78) --
	( 40.13, 81.87) --
	( 51.79, 81.87) --
	( 51.79, 86.78) --
	( 40.13, 86.78) --
	cycle;

\path[draw=drawColor,line width= 1.1pt,line join=round] ( 40.13, 81.87) -- ( 51.79, 81.87);
\definecolor{fillColor}{gray}{0.20}

\path[draw=drawColor,line width= 0.4pt,line join=round,line cap=round,fill=fillColor] ( 58.91, 72.05) circle (  0.68);

\path[draw=drawColor,line width= 0.4pt,line join=round,line cap=round,fill=fillColor] ( 58.91, 72.05) circle (  0.68);

\path[draw=drawColor,line width= 0.4pt,line join=round,line cap=round,fill=fillColor] ( 58.91, 96.60) circle (  0.68);

\path[draw=drawColor,line width= 0.6pt,line join=round] ( 58.91, 86.78) -- ( 58.91, 91.69);

\path[draw=drawColor,line width= 0.6pt,line join=round] ( 58.91, 81.87) -- ( 58.91, 76.96);
\definecolor{fillColor}{RGB}{217,95,2}

\path[draw=drawColor,line width= 0.6pt,line join=round,line cap=round,fill=fillColor] ( 53.08, 86.78) --
	( 53.08, 81.87) --
	( 64.74, 81.87) --
	( 64.74, 86.78) --
	( 53.08, 86.78) --
	cycle;

\path[draw=drawColor,line width= 1.1pt,line join=round] ( 53.08, 81.87) -- ( 64.74, 81.87);

\path[draw=drawColor,line width= 0.6pt,line join=round] ( 80.50, 96.60) -- ( 80.50,111.34);

\path[draw=drawColor,line width= 0.6pt,line join=round] ( 80.50, 86.78) -- ( 80.50, 76.96);
\definecolor{fillColor}{RGB}{27,158,119}

\path[draw=drawColor,line width= 0.6pt,line join=round,line cap=round,fill=fillColor] ( 74.67, 96.60) --
	( 74.67, 86.78) --
	( 86.33, 86.78) --
	( 86.33, 96.60) --
	( 74.67, 96.60) --
	cycle;

\path[draw=drawColor,line width= 1.1pt,line join=round] ( 74.67, 91.69) -- ( 86.33, 91.69);
\definecolor{fillColor}{gray}{0.20}

\path[draw=drawColor,line width= 0.4pt,line join=round,line cap=round,fill=fillColor] ( 93.45,111.34) circle (  0.68);

\path[draw=drawColor,line width= 0.4pt,line join=round,line cap=round,fill=fillColor] ( 93.45, 76.96) circle (  0.68);

\path[draw=drawColor,line width= 0.4pt,line join=round,line cap=round,fill=fillColor] ( 93.45, 76.96) circle (  0.68);

\path[draw=drawColor,line width= 0.4pt,line join=round,line cap=round,fill=fillColor] ( 93.45,101.51) circle (  0.68);

\path[draw=drawColor,line width= 0.4pt,line join=round,line cap=round,fill=fillColor] ( 93.45,101.51) circle (  0.68);

\path[draw=drawColor,line width= 0.4pt,line join=round,line cap=round,fill=fillColor] ( 93.45,101.51) circle (  0.68);

\path[draw=drawColor,line width= 0.6pt,line join=round] ( 93.45, 91.69) -- ( 93.45, 96.60);

\path[draw=drawColor,line width= 0.6pt,line join=round] ( 93.45, 86.78) -- ( 93.45, 81.87);
\definecolor{fillColor}{RGB}{217,95,2}

\path[draw=drawColor,line width= 0.6pt,line join=round,line cap=round,fill=fillColor] ( 87.62, 91.69) --
	( 87.62, 86.78) --
	( 99.28, 86.78) --
	( 99.28, 91.69) --
	( 87.62, 91.69) --
	cycle;

\path[draw=drawColor,line width= 1.1pt,line join=round] ( 87.62, 91.69) -- ( 99.28, 91.69);
\definecolor{fillColor}{gray}{0.20}

\path[draw=drawColor,line width= 0.4pt,line join=round,line cap=round,fill=fillColor] (115.04,130.98) circle (  0.68);

\path[draw=drawColor,line width= 0.6pt,line join=round] (115.04,111.34) -- (115.04,121.16);

\path[draw=drawColor,line width= 0.6pt,line join=round] (115.04,100.29) -- (115.04, 86.78);
\definecolor{fillColor}{RGB}{27,158,119}

\path[draw=drawColor,line width= 0.6pt,line join=round,line cap=round,fill=fillColor] (109.21,111.34) --
	(109.21,100.29) --
	(120.86,100.29) --
	(120.86,111.34) --
	(109.21,111.34) --
	cycle;

\path[draw=drawColor,line width= 1.1pt,line join=round] (109.21,106.43) -- (120.86,106.43);

\path[draw=drawColor,line width= 0.6pt,line join=round] (127.99,106.43) -- (127.99,121.16);

\path[draw=drawColor,line width= 0.6pt,line join=round] (127.99, 96.60) -- (127.99, 86.78);
\definecolor{fillColor}{RGB}{217,95,2}

\path[draw=drawColor,line width= 0.6pt,line join=round,line cap=round,fill=fillColor] (122.16,106.43) --
	(122.16, 96.60) --
	(133.82, 96.60) --
	(133.82,106.43) --
	(122.16,106.43) --
	cycle;

\path[draw=drawColor,line width= 1.1pt,line join=round] (122.16,101.51) -- (133.82,101.51);

\path[draw=drawColor,line width= 0.6pt,line join=round] (149.57,126.07) -- (149.57,140.80);

\path[draw=drawColor,line width= 0.6pt,line join=round] (149.57,111.34) -- (149.57, 91.69);
\definecolor{fillColor}{RGB}{27,158,119}

\path[draw=drawColor,line width= 0.6pt,line join=round,line cap=round,fill=fillColor] (143.74,126.07) --
	(143.74,111.34) --
	(155.40,111.34) --
	(155.40,126.07) --
	(143.74,126.07) --
	cycle;

\path[draw=drawColor,line width= 1.1pt,line join=round] (143.74,121.16) -- (155.40,121.16);
\definecolor{fillColor}{gray}{0.20}

\path[draw=drawColor,line width= 0.4pt,line join=round,line cap=round,fill=fillColor] (162.53, 86.78) circle (  0.68);

\path[draw=drawColor,line width= 0.4pt,line join=round,line cap=round,fill=fillColor] (162.53,135.89) circle (  0.68);

\path[draw=drawColor,line width= 0.6pt,line join=round] (162.53,117.47) -- (162.53,130.98);

\path[draw=drawColor,line width= 0.6pt,line join=round] (162.53,106.43) -- (162.53, 96.60);
\definecolor{fillColor}{RGB}{217,95,2}

\path[draw=drawColor,line width= 0.6pt,line join=round,line cap=round,fill=fillColor] (156.70,117.47) --
	(156.70,106.43) --
	(168.35,106.43) --
	(168.35,117.47) --
	(156.70,117.47) --
	cycle;

\path[draw=drawColor,line width= 1.1pt,line join=round] (156.70,111.34) -- (168.35,111.34);

\path[draw=drawColor,line width= 0.6pt,line join=round] (184.11,145.71) -- (184.11,170.26);

\path[draw=drawColor,line width= 0.6pt,line join=round] (184.11,126.07) -- (184.11, 96.60);
\definecolor{fillColor}{RGB}{27,158,119}

\path[draw=drawColor,line width= 0.6pt,line join=round,line cap=round,fill=fillColor] (178.28,145.71) --
	(178.28,126.07) --
	(189.94,126.07) --
	(189.94,145.71) --
	(178.28,145.71) --
	cycle;

\path[draw=drawColor,line width= 1.1pt,line join=round] (178.28,135.89) -- (189.94,135.89);

\path[draw=drawColor,line width= 0.6pt,line join=round] (197.06,135.89) -- (197.06,160.44);

\path[draw=drawColor,line width= 0.6pt,line join=round] (197.06,115.02) -- (197.06, 96.60);
\definecolor{fillColor}{RGB}{217,95,2}

\path[draw=drawColor,line width= 0.6pt,line join=round,line cap=round,fill=fillColor] (191.23,135.89) --
	(191.23,115.02) --
	(202.89,115.02) --
	(202.89,135.89) --
	(191.23,135.89) --
	cycle;

\path[draw=drawColor,line width= 1.1pt,line join=round] (191.23,126.07) -- (202.89,126.07);
\end{scope}
\begin{scope}
\definecolor{drawColor}{gray}{0.30}

\node[text=drawColor,anchor=base east,inner sep=0pt, outer sep=0pt, scale=  0.88] at ( 26.76, 69.02) {0};

\node[text=drawColor,anchor=base east,inner sep=0pt, outer sep=0pt, scale=  0.88] at ( 26.76, 93.57) {5};

\node[text=drawColor,anchor=base east,inner sep=0pt, outer sep=0pt, scale=  0.88] at ( 26.76,118.13) {10};

\node[text=drawColor,anchor=base east,inner sep=0pt, outer sep=0pt, scale=  0.88] at ( 26.76,142.68) {15};

\node[text=drawColor,anchor=base east,inner sep=0pt, outer sep=0pt, scale=  0.88] at ( 26.76,167.23) {20};
\end{scope}
\begin{scope}
\definecolor{drawColor}{gray}{0.30}

\node[text=drawColor,anchor=base,inner sep=0pt, outer sep=0pt, scale=  0.88] at ( 52.44, 56.13) {3};

\node[text=drawColor,anchor=base,inner sep=0pt, outer sep=0pt, scale=  0.88] at ( 86.97, 56.13) {4};

\node[text=drawColor,anchor=base,inner sep=0pt, outer sep=0pt, scale=  0.88] at (121.51, 56.13) {5};

\node[text=drawColor,anchor=base,inner sep=0pt, outer sep=0pt, scale=  0.88] at (156.05, 56.13) {6};

\node[text=drawColor,anchor=base,inner sep=0pt, outer sep=0pt, scale=  0.88] at (190.59, 56.13) {7};
\end{scope}
\begin{scope}
\definecolor{drawColor}{RGB}{0,0,0}

\node[text=drawColor,anchor=base,inner sep=0pt, outer sep=0pt, scale=  1.10] at (121.51, 44.09) {System dimension};
\end{scope}
\begin{scope}
\definecolor{drawColor}{RGB}{0,0,0}

\node[text=drawColor,rotate= 90.00,anchor=base,inner sep=0pt, outer sep=0pt, scale=  1.10] at ( 13.08,121.16) {SHD};
\end{scope}
\begin{scope}
\definecolor{drawColor}{RGB}{0,0,0}

\node[text=drawColor,anchor=base west,inner sep=0pt, outer sep=0pt, scale=  1.10] at ( 27.70, 14.44) {Test type};
\end{scope}
\begin{scope}
\definecolor{drawColor}{gray}{0.20}

\path[draw=drawColor,line width= 0.6pt,line join=round,line cap=round] ( 85.70, 12.45) --
	( 85.70, 14.61);

\path[draw=drawColor,line width= 0.6pt,line join=round,line cap=round] ( 85.70, 21.84) --
	( 85.70, 24.01);
\definecolor{fillColor}{RGB}{27,158,119}

\path[draw=drawColor,line width= 0.6pt,line join=round,line cap=round,fill=fillColor] ( 80.28, 14.61) rectangle ( 91.12, 21.84);

\path[draw=drawColor,line width= 0.6pt,line join=round,line cap=round] ( 80.28, 18.23) --
	( 91.12, 18.23);
\end{scope}
\begin{scope}
\definecolor{drawColor}{gray}{0.20}

\path[draw=drawColor,line width= 0.6pt,line join=round,line cap=round] (152.70, 12.45) --
	(152.70, 14.61);

\path[draw=drawColor,line width= 0.6pt,line join=round,line cap=round] (152.70, 21.84) --
	(152.70, 24.01);
\definecolor{fillColor}{RGB}{217,95,2}

\path[draw=drawColor,line width= 0.6pt,line join=round,line cap=round,fill=fillColor] (147.28, 14.61) rectangle (158.12, 21.84);

\path[draw=drawColor,line width= 0.6pt,line join=round,line cap=round] (147.28, 18.23) --
	(158.12, 18.23);
\end{scope}
\begin{scope}
\definecolor{drawColor}{RGB}{0,0,0}

\node[text=drawColor,anchor=base west,inner sep=0pt, outer sep=0pt, scale=  0.88] at ( 98.42, 15.20) {First order};
\end{scope}
\begin{scope}
\definecolor{drawColor}{RGB}{0,0,0}

\node[text=drawColor,anchor=base west,inner sep=0pt, outer sep=0pt, scale=  0.88] at (165.42, 15.20) {Second order};
\end{scope}
\end{tikzpicture}
   \caption{Structural Hamming Distances (SHD) between the true graph that simulated data and the graphs estimated by using either first or second order tests in the experiment in \cref{subsec:sim-eca}.}
  \label{fig:eca}
\end{figure}
Constrained based causal learning algorithms, such as the CA-algorithm, estimate the causal graph by sequentially testing local independence $j\not\rightarrow k\mid C$ for nodes $j,k$ given conditioning sets $C\subset V\setminus\{j\}$ of increasing size. 
If at some point, a local independence $j\not\rightarrow k\mid C$ is found, the edge $j\rightarrow k$ is removed from the graph. 

For each simulated Hawkes process, we run the CA-algorithm using either the first- or the second-order tests and obtain a resulting estimated graph. We then compare the estimated graphs to the true graph that generated the Hawkes process by the Structural Hamming Distance (SHD), which measures the number of edge additions, removals or flips that is needed to convert the estimated graph into the true graph. That is, the SHD measures how far the estimated graph is from the true graph. 
\cref{fig:eca} shows the resulting Structural Hamming Distances for the different dimensions. We observe that for all dimensions, the second-order approach performs as well or better than the first-order approach. Notably, this is more outspoken as dimensions increase: In larger systems, more processes are marginalized away when testing $j\not\rightarrow k \mid C$, and so the effect of model misspecification is more severe for larger dimensions. 

\section{Neuron firing data}\label{sec:real-world}
We employ a causal discovery algorithm using our proposed tests to a data set of neuron firing in turtles.\footnote{Data provided by Associate Professor Rune W. Berg, University of Copenhagen.} The turtles were exposed to a stimuli in a period of $10$ seconds, in which the activity of $d=6$ channels were measured. The experiment was repeated $5$ times. 

For each repetition, we employ the Causal Analysis (CA) algorithm from \cite{Meek:2014} to learn the causal structure, using either first- or second-order tests. \cref{fig:neuron_firing} shows data from the first repetition of the experiment and the resulting learned graphs (repetitions 2--5 are shown in \cref{fig:rep-2,fig:rep-3,fig:rep-4,fig:rep-5} in the appendix). The graph estimated using second-order tests is sparser than the one using first-order tests. This concurs with our motivation for including second-order terms: when level is lost due to misspecification, the edge $j \rightarrow k$ will too often remain in the graph, even though $j \not\rightarrow k \mid C$ for some $C$. Using first order tests results in a denser and less informative graph.
This effect is more outspoken in the neuron firing data than in the simulated data in \cref{sec:simulation}: While the synthetic data was truly simulated from a Hawkes process, and so the misspecification would only be due to marginalization, there may be additional misspecification in the real data if the full process is not truly a Hawkes process. 

Since ground truth graphs for the neural connections are not available, we cannot directly evaluate which test provides estimated graphs closer to ground truth. Instead, we compare the first- and second-order tests by their consistency across the $5$ repetitions, i.e. how similar the estimated graphs are from the $5$ repetitions. 
For each repetition, a separate graph is learned using the CA-algorithm, with a test using either first- or second-order terms. 
In \cref{tab:consistency} we display the proportion of edges where either \textit{i)} all $5$ graphs agree on the presence or absence of the edges and \textit{ii)} at least $4$ of $5$ graphs agree. 
As a baseline, we include the theoretical proportions, if in each graph, an edge would appear randomly with a probability of $1/2$. Self-edges, which are easy to detect, and hence inflates consistency, are excluded from all numbers.
We observe that the second-order approach is more consistent in terms of both agreement between all $5$ repetitions and agreement between at least $4$ repetitions. 
\begin{figure}[t]
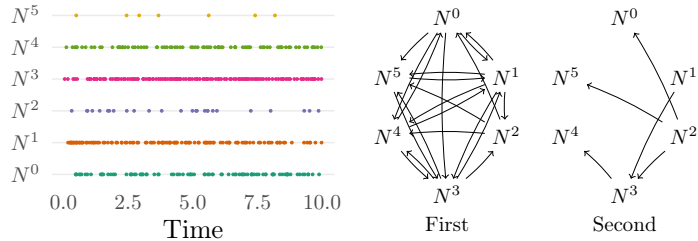

\centering
    \subfloat{
    \resizebox{0.3\linewidth}{!}{%
}
    }
    \caption{(\textit{left}) The first repetition of the experiment. Each point corresponds to one neuron firing. 
    (\textit{right}) Output of the CA algorithm on the first repetition, when the test of local independence either uses a first-order test (`First') or a second-order test (`Second'). \Cref{fig:rep-2,fig:rep-3,fig:rep-4,fig:rep-5} in the appendix show the similar plots and graphs for the repetitions two through five. 
    }
    \label{fig:neuron_firing}
\end{figure}
\begin{table}[t]
\centering
\caption{
Consistency of estimated graphs from the $5$ repetitions of the stimulus experiment. 
}
\resizebox{\linewidth}{!}{%
\begin{tabular}{llll}\hline
                                          & Baseline & First    & Second   \\ \hline
Edges consistent in all 5 repetitions     & $6.25\%$ & $23.3\%$ & $\mathbf{26.7\%}$ \\
Edges consistent in a least 4 repetitions & $37.5\%$ & $40.0\%$ & $\mathbf{56.7\%}$ \\ 
Number of edges present                   & $50.0\%$ & $62.7\%$ & $30.7\%$ \\ \hline
\end{tabular}%
}
\label{tab:consistency}
\end{table}

\section{Discussion}
In this paper, we formulated a framework for testing local independence in point processes. We introduced a test of local independence that fits intensities using basis expansions and tests the local independence hypothesis $j \not\rightarrow k \mid C$ by testing significance of contributions of the process $N^j$ to the intensity $\lambda^{k, C\cup\{j\}}$. 

We addressed the issue of marginalization: Even if the full data generating mechanism is a known and simple model class, such as Hawkes processes, a partially observed system with some event types unobserved cannot necessarily be modelled as a Hawkes process. This issue is native to (conditional) local independence testing, since the local independence $j \not\rightarrow k\mid C$ relates to the marginal distribution of $N^{\{j,k\}\cup C}$. 
To overcome this misspecification, we proved that, when facing marginalized variables, the intensity can be arbitrarily well approximated by expansions in terms of iterated integrals, and we have verified that including higher-order interactions leads to an improved level of the test of $j \not\rightarrow k \mid C$. 

The availability of an empirical local independence test is quintessential to constraint-based causal structure learning algorithms for point processes, and we have validated in simulation studies that using our proposed test, one can from data obtain good estimates of the underlying graph. 
We applied our approach to a real-world data set on neuron spiking in turtles, and found that including higher-order interactions resulted in sparser, more informative estimated networks.

\bibliographystyle{IEEEtran}
\bibliography{bibliography}

\clearpage
\appendix
\section{Appendix}
\subsection{Proof of Theorem \ref{thm:representation_0}}
In this appendix, we prove \cref{thm:representation_0}. The proof first shows the result for a univariate process ($|C| = 1$), and then argues that result can easily be extended to the multivariate setting. 

We stress that the motivation for the theorem is to show convergence of the representation. In practice, many other kernel functions than those appearing in the proof, could also be used to describe the system, and so our interest lies very little in the concrete functional forms used. 

Let $\tau_1, \tau_2, \ldots$ be the jumps of $N$ starting at $0$ and moving backwards in time. That is $\ldots < \tau_2 < \tau_1 < 0$. 
\begin{dfn}
	For $s < 0$, let $\mathcal{F}_s$ denote the $\sigma$-algebra generated by events in $[s, 0)$. That is $$\mathcal{F}_s = \sigma(\tau_1 \lor s, \tau_2 \lor s, \ldots),$$
	where $\tau\lor s = \max(\tau,s)$. 
	Define also $\mathcal{F}_{0-} = \sigma(\cup_{s < 0} \mathcal{F}_s)$. 
\end{dfn}
\begin{prop}\label{rep:prop_generator_sigma_alg}
	$\mathcal{F}_0$ equals $\mathcal{F}_{-\infty}:= \sigma(\tau_1, \tau_2, \ldots)$. 
\end{prop}
\begin{proof}
	For all $i$, $\tau_i \lor s$ is $\sigma(\tau_i)$-measurable, and in particular, $\mathcal{F}_{-\infty}$-measurable. Therefore $\mathcal{F}_{s} = \sigma(\tau_1 \lor s, \ldots) \subseteq \mathcal{F}_{-\infty}$ and so $\mathcal{F}_{0-} = \sigma(\cup_{s<0} \mathcal{F}_s) \subseteq \mathcal{F}_{-\infty}$. 
	
	Reversely, $\tau_n$ is $\mathcal{F}_0$-measurable for each $n$. $\mathcal{F}_{-\infty}$ is the smallest $\sigma$-algebra making all $\tau_n$'s measurable, so $\mathcal{F}_{-\infty} \subseteq \mathcal{F}_0$ will follow. To see that $\tau_n$ is $\mathcal{F}_0$-measureable, consider any $n$. $(\tau_n \lor s) \rightarrow \tau_n$ for $s \rightarrow -\infty$ (potentially with $\tau_n = -\infty)$, and so since $(\tau_n \lor s)$ is $\mathcal{F}_0$-measurable for each $s$, $\tau_n$ is $\mathcal{F}_0$-measurable. 
\end{proof}

\begin{prop}\label{prop:dense_subspace}
	The union of function spaces $\cup_{s < 0} \mathcal{L}^1(\mathcal{F}_s)$ is dense in $\mathcal{L}^1(\mathcal{F}_{0-})$. 
\end{prop}
\begin{proof}
	Take any $\lambda \in \mathcal{L}^1(\mathcal{F}_{0-})$. By the tower property, $\lambda_s := E[\lambda \mid \mathcal{F}_s] \in \mathcal{L}^1(\mathcal{F}_s)$ and further from the martingale convergence theorem, $(\lambda_s)_{s < 0}$ is a martingale (in $-s$) and $E[\lambda \mid \mathcal{F}_s]$ converges in $\mathcal{L}^1$ to $E[\lambda \mid \mathcal{F}_{0-}] = \lambda$ as $s \to -\infty$. 
	
	Because each $\lambda_s \in \mathcal{L}^1(\mathcal{F}_s) \subseteq \cup_{s<0} \mathcal{L}^1(\mathcal{F}_s)$, it follows that $\cup_s \mathcal{L}^1(\mathcal{F}_s)$ is dense in $\mathcal{L}^1(\mathcal{F}_{0-})$. 
\end{proof}

We now show that for any $\lambda \in \mathcal{L}^1(\mathcal{F}_s)$ and for each $M \in \mathbb{N}$ that $\lambda 1_{N([s, 0)) = M}$ can be written as a sum of integrals of deterministic functions. These integrands will play a role similar to Volterra kernels, but only given the count $N([s, 0))$. We then sum over these terms, to obtain a general representation of $\lambda$. 

It is well known that if $Y \in \mathcal{L}^1(\sigma(X_1, X_2, \ldots))$ for some random variables $X_1, \ldots$, then there exists a measurable map $f$ such that $Y = f(X_1, X_2, \ldots)$. In the case of event times truncated at $s$, $\tau_n \lor s$, this corresponds to that if $\lambda \in \mathcal{L}^1(\mathcal{F}_s)$ there exists a function $f$ such that $$\lambda = f(\tau_1 \lor s, \tau_2 \lor s, \ldots)$$

To obtain an integral representation of $\lambda$, we can utilize this function. 
Define $f_s^n(t_1, \ldots, t_n) = f(t_1, \ldots, t_n, s, s, \ldots)$ as the evaluation of $f$ in $(t_1, \ldots, t_n)$ and then the $s$ in all other entries of the function. We will write $f^n$ if $s$ is clear from the context or even $f(t_1, \ldots, t_n)$.

As a motivation for the below proof, suppose that we knew that exactly one event occurred in the interval $A := [s,0)$, i.e. $\tau_1 \in A, \tau_n \notin A$ for $n\geq 2$. Then one could write:
\begin{align*}
	\lambda &= f(\tau_1\lor s, \tau_2\lor s,\ldots) = f(\tau_1, s, s, \ldots) \\
	&= \int_s^{0-} f(t, s, s, \ldots) N(\dd t) = \int_s^{0-} f^1 N(\dd t)
\end{align*}
This however depends heavily on the assumption that $N(A) = 1$. If instead the interval contained $m$ events, then $\int_s^{0-} f^1(t) N(\dd t) = f^1(\tau_1) + \ldots + f^1(\tau_m)$ which is not equal to $\lambda$ (because in this case $\lambda = f^m(\tau_1, \ldots, \tau_m)$). 

The following proposition devices a procedure, such that one can obtain $f(\tau_1)$ exactly if $N(A) = 1$ and else $0$, using only integrals of deterministic functions. For a function $h(t_1, \ldots, t_n)$, we use the shorthand notation
\begin{align*}
    \int_A h \dd N(t^n) := \int_A\cdots\int_A h(t_1, \ldots, t_n) \dd N(t_1) \cdots \dd N(t_n). 
\end{align*}

\begin{prop}\label{rep:prop_compensation_scheme_t1}
	Assume $N$ is a simple, non-exploding point process. Let $\lambda \in \mathcal{L}^1(\mathcal{F}_s)$ and $A = [s, 0)$.
	Then 
	\begin{align}\label{rep:eq_compensation_scheme_t1}
		\sum_{n=1}^L \beta_n \int_A f(t_1) 1_{D_n} \dd N(t^n) \stackrel{\textrm{a.s.}}{\longrightarrow} \lambda 1_{N(A) = 1} 
	\end{align}
	for $L \longrightarrow \infty$ where $\beta_n = \frac{(-1)^{n-1}}{(n-1)!}, n \geq 1$, and:
	\begin{align}
		D_n = \{(t_1, \ldots, t_n) \in [-s, 0)^n\,\mid \, t_i \neq t_j \text{ for }i \neq j\}\notag
	\end{align}
\end{prop}

\begin{proof}
	Observe that while we integrate over sequences $(t_1, \ldots, t_n)$, we evaluate only the function $f^1(t_1)$ in $t_1$. The indicator function $1_{D_n}$ still is evaluated in $(t_1, \ldots, t_n)$.
	For this reason
	\begin{align*}
	&\int_A f(t_1) 1_{D_n}(t_1, \ldots, t_n) \dd N(t^n) \\
	&= \left[\int_A f(t_1) \dd N(t_1)\right]\begin{pmatrix}N(A) - 1 \\ n - 1\end{pmatrix}(n-1)!
	\end{align*}
    This follows because for each event time $\tau \in A$, there are exactly $\begin{pmatrix}N(A) - 1 \\ n - 1\end{pmatrix}(n-1)!$ tuples 	$(\tau, t_2, \ldots, t_n)$ where $\tau$ is the first element and no elements are identical. 

It then follows that: 
\begin{align*}
	&\sum_{n=1}^{N(A)} \beta_n \int_A f(t_1) 1_{D_n} \dd N(t^n) \\
	&= \left[\int_A f(t) \dd N(t)\right] \sum_{n=1}^{N(A)} (-1)^{n-1} \begin{pmatrix}N(A)-1\\ n-1\end{pmatrix} \\
	&= \begin{cases}
		f(\tau_1) & N(A) = 1 \\
		0 & \textrm{else}
	\end{cases} = \lambda 1_{N(A) = 1}
\end{align*}
This last step is utilizes that for $M = 1$, 
$\sum_{n=1}^M (-1)^{n-1} \binom{M-1}{n-1} = 1$, and for $M > 1$, the binomial formula implies that
\begin{align*}
    0 &= \left(1 + (-1)\right)^{M-1} \\
    &=\sum_{n=0}^{M-1} (-1)^n \begin{pmatrix}M-1 \\ n\end{pmatrix} \\
    &=\sum_{n=1}^{M} (-1)^{n-1} \begin{pmatrix}M-1 \\ n-1\end{pmatrix}.
\end{align*}
Since the integrand $1_{D_n}$ is $0$ for $n \geq N(A)$, and $P\left(N(A) < \infty\right) = 1$, it follows that 
$$
\sum_{n=1}^{L} \beta_n \int_A f(t_1) 1_{D_n} \dd N(t^n) \stackrel{\textrm{a.s.}}{\longrightarrow} \lambda 1_{N(A) = 1} \textrm{ for } L \rightarrow \infty
$$
\end{proof}
This extends to the following corollary:
\begin{cor}\label{rep:cor_compensation_scheme_t_n}
	Let $\lambda\in \mathcal{L}^1(\mathcal{F}_s)$. For $M \in \mathbb{N}$, one has:
	\begin{align}
		\sum_{n=M}^L \beta_n^M \int_A f(t_1, \ldots, t_M) 1_{D_n} 1_{O_M} \dd N(t^{n}) \stackrel{\textrm{a.s.}}{\longrightarrow} \lambda 1_{N(A) = M} \notag
	\end{align}
	with $\beta_n^M = \frac{(-1)^{n-M}}{(n-M)!}$ for $n \geq M$ and
	\begin{align}
		O_M = \{(t_1, \ldots, t_M) \in [-s, 0)^n\,\mid \, t_1 < t_2 \ldots < t_M \}\notag
	\end{align}
\end{cor}
\begin{proof}
	The case $M = 1$ is covered in \cref{rep:prop_compensation_scheme_t1}. For $M \geq 2$, the result essentially is the same, with the additional requirement that the first $M$ jumps should be ordered, which is handled by $1_{O_M}$.

	Apart from this, combinatorics of how many tuples $(t_1, \ldots, t_n)$ with $t_1 < \ldots < t_M$ ordered (as fixed by $O_M$) and all $t$'s distinct (by $D_n$) remains the same, in particular
	\begin{align*}
	&\int_A f(t_1, \ldots, t_M) 1_{D_n} 1_{O_M} \dd N(t^n)\\
	&= \left[\int_A f(t_1, \ldots, t_M) 1_{O_M} \dd N(t^M)\right]\begin{pmatrix}N(A) - M \\ n - M\end{pmatrix}(n-M)!
	\end{align*}
	Consequently, the proof from \cref{rep:prop_compensation_scheme_t1} also applies in the case of $1_{N(A) = M}$. 
\end{proof}
Extending further on \cref{rep:prop_compensation_scheme_t1} and Corollary \ref{rep:cor_compensation_scheme_t_n}, we may include the base-rate $\lambda 1_{N(A) = 0}$. Let $h^0$ be the value of $\lambda$ on the set $\{N(A) = 0\}$ (that is $h^0 = f(s, s, \ldots)$). Now $\sum_{n=1}^L \int_A \left[f(t_1) - h^0 \right] 1_{D_n} \dd N(t^{n})$ will return the \textit{additional to base-rate} intensity $f(\tau_1) - h^0$ if $N(A) = 1$ and $0$ else.
We combine the above:
\begin{prop}\label{rep:prop_compensation_scheme_any_number}
	Assume $N$ is a non-exploding point process, and assume $\lambda \in \mathcal{L}^1(\mathcal{F}_s)$. Then
	\begin{equation*}
		\sum_{M=1}^L \sum_{n=M}^L \beta_n^M \int_A f(t_1, \ldots, t_M) 1_{D_n} 1_{O_M} \dd N(t^n)\stackrel{\textrm{a.s.}}{\longrightarrow}\lambda,
	\end{equation*}
	for $L \longrightarrow \infty$. 
\end{prop}
\begin{proof}
	As above, the almost sure convergence follows simply by decomposing $\lambda = \lambda 1_{N(A) = 0} + \sum_{M \in \mathbb{N}} \lambda 1_{N(A) = M}$, and again observing that since $P(N(A) < \infty)$, for every $\omega$, the left hand side will arrive at the true value for some finite $L$. 
\end{proof}

Finally we are able to prove the main result.
\begin{proof}[Proof of \cref{thm:representation_0}]
Observe that each function $\beta_n^M f 1_{D_n} 1_{O_M} \in E_n$, and 
so $$\sum_{M=1}^L \sum_{n=M}^L \beta_n^M \int_A f(t_1, \ldots, t_M) 1_{D_n} 1_{O_M} \dd N(t^n)$$ is in $\bigoplus_{n=0}^L W_n$. 
Be reminded that by \cref{prop:dense_subspace}, $\cup_{s < 0}\mathcal{L}^1(\mathcal{F}_s)$ is ($\mathcal{L}^1$-)dense in $\mathcal{L}^1(\mathcal{F}_{0-})$, and for every element $\lambda$ of $\cup_{s< 0}\mathcal{L}^1(\mathcal{F}_s)$ there exists a sequence in $\bigoplus_{n \in \mathbb{N}} W_n$ converging almost surely to $\lambda$. Consequently, as both $\mathcal{L}^1$ and almost sure convergence implies convergence in probability, for any $\lambda \in \mathcal{L}^1(\mathcal{F}_{0-})$ there exist a sequence in $\bigoplus_{n\in\mathbb{N}} W_n$ converging to $\lambda$ in probability. 

Consider now any $\lambda \in \mathcal{L}(\mathcal{F}_{0-}^C)$ with $|\lambda| < \infty$ a.s. Trivially $\lambda_k := 1_{|\lambda| < K} \lambda$ converges in probability to $\lambda$ for $k\to\infty$. Further each $\lambda_k \in \L^1(\mathcal{F}_{0-})$, and hence there exists a sequence there exists a sequence in $\bigoplus_{n \in \mathbb{N}} W_n$ converging almost surely to $\lambda$, completing the proof in the case without marks. 
\end{proof}
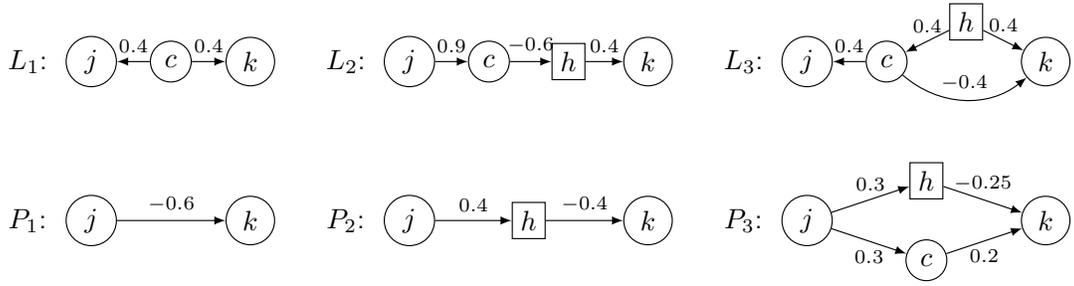
\begin{figure*}[t]
    \centering
    \resizebox{0.95\textwidth}{!}{%
		\begin{tikzpicture}[>=latex,font=\sffamily,every node/.style={inner sep = 3pt}]
			\node[draw, circle] (A) at (4, 0) {$j$};
		    \node[draw, circle] (C) at (5,0) {$c$};
		    \node[draw, circle] (B) at (6,0) {$k$};
		    \path [->] (C) edge node[above] {\scriptsize $0.4$} (A);
		    \path [->] (C) edge node[above] {\scriptsize$0.4$} (B);

			\node[draw, circle] (A) at (8, 0) {$j$};
		    \node[draw, circle] (C1) at (9,0) {$c$};
		    \node[draw] (C2) at (10,0) {$h$};
		    \node[draw, circle] (B) at (11,0) {$k$};
		    \path [->] (A) edge node[above] {\scriptsize $0.9$} (C1);
		    \path [->] (C1) edge node[above] {\scriptsize $-0.6$} (C2);
		    \path [->] (C2) edge node[above] {\scriptsize $0.4$} (B);
		    \node[draw, circle] (A) at (13, 0) {$j$};
		    \node[draw, circle] (C1) at (14,0) {$c$};
		    \node[draw] (C2) at (15,0.5) {$h$};
		    \node[draw, circle] (B) at (16,0) {$k$};
		    \path [->] (C1) edge node[above] {\scriptsize $0.4$} (A);
		    \path [->] (C2) edge node[above] {\scriptsize $0.4$} (C1);
		    \path [->] (C2) edge node[above] {\scriptsize $0.4$} (B);
		    \path [->] (C1) edge[bend right=40] node[above] {\scriptsize $-0.4$} (B);
		    \node at (3.2, 0) {$L_1$:};
		    \node at (7.2, 0) {$L_2$:};
		    \node at (12.2, 0) {$L_3$:};

			\node[draw, circle] (A) at (4, -2) {$j$};
		    \node[draw, circle] (B) at (6, -2) {$k$};
		    \path [->] (A) edge node[above] {\scriptsize $-0.6$} (B);
		    \node[draw, circle] (A) at (8, -2) {$j$};
		    \node[draw] (C2) at (9.5,-2) {$h$};
		    \node[draw, circle] (B) at (11, -2) {$k$};
		    \path [->] (A) edge node[above] {\scriptsize $0.4$} (C2);
		    \path [->] (C2) edge node[above] {\scriptsize $-0.4$} (B);

			\node[draw, circle] (A) at (13, -2) {$j$};
		    \node[draw, circle] (C1) at (14.5, -2.5) {$c$};
		    \node[draw] (C2) at (14.5, -1.5) {$h$};
		    \node[draw, circle] (B) at (16, -2) {$k$};
		    \path [->] (A) edge node[below] {\scriptsize $0.3$} (C1);
		    \path [->] (A) edge node[above] {\scriptsize $0.3$} (C2);
		    \path [->] (C1) edge node[below] {\scriptsize $0.2$} (B);
		    \path [->] (C2) edge node[above] {\scriptsize $-0.25$} (B);

		    \node at (3.2, -2) {$P_1$:};
		    \node at (7.2, -2) {$P_2$:};
		    \node at (12.2, -2) {$P_3$:};
		\end{tikzpicture}}
    \caption{Simulation parameters for the experiment in \cref{subsec:level-power}}
    \label{fig:structures-details}
\end{figure*}
The above framework is readily extended to marked point processes. Remember that with $V = \{1, \ldots, d\}$ and $C \subseteq V$, one has for any Borel measurable set $A$ that:
\begin{align}
	N(A\times C) = \sum_{v\in C}N(A\times\{v\}) = \sum_{v\in C} N^{v}(A)\notag
\end{align}
When integrating, this factorizes:
\begin{align*}
	&\int_{A\times C} f(x, v) N(\dd x, \dd v) \\
	&=  \int_{A} f(x, v) \sum_{v \in C} N^{v}(\dd x)\\
	&= \sum_{v\in C}\int_A f^{v}(x) N^{v}(\dd x)
\end{align*}
where we let $f^{v}(x) := f(x, v)$. Similarly in higher dimensions:
\begin{align*}
	&\int_{A\times C} f(x_1, v_1, \ldots, x_n, v_n) N(\dd x^n \times \dd v^n)\\
	&= \sum_{|\alpha|= n}\int_{A\times C}  f^\alpha(x_1, \ldots, x_n) \underbrace{N^{\alpha_1}(x^1)\cdots N^{\alpha_n}(x^n)}_{=: N^\alpha(\dd x^n)}
\end{align*}
where $f^\alpha(x_1, \ldots, x_n) = f(x_1, \alpha_1, \ldots x_n, \alpha_n)$ and $\alpha \in V^n$ is some tuple of length $n$.

Thus, the combinatorics of the one-dimensional case apply also in the marked setting and the result thus directly transfers to the multi-dimensional case:
In the marked setting, the generated $\sigma$-field becomes $\mathcal{F}_s = \sigma((\tau_1 \lor s, v_1 1_{\tau_1 > s}), \ldots )$. A multivariate version of \cref{rep:prop_generator_sigma_alg} follows because $(\tau_1\lor s, v_11_{\tau_1>s}) \rightarrow (\tau_1, v_11_{\tau_1 > -\infty})$\footnote{Which is the desired limit, with the convention  that $v_n = 0$ if $\tau_n = -\infty$.}, and so denseness of $\cup_s \mathcal{L}^1(\mathcal{F}_s)$ also follows in the marked case.
	Thus the function $f$ could have been written:
\begin{align}
	\lambda = f\left((\tau_1 \lor s, v_1 1_{\tau_1 > s}), \ldots\right)\notag
\end{align}
In \cref{rep:prop_compensation_scheme_t1}, one could have proceeded in exactly the same way, but using integrals $\int_{A\times V} f(t_1, v_1) N(\dd t_1 \times v_1)$ instead.

Therefore also \cref{rep:prop_compensation_scheme_any_number} generalizes such that any $\lambda \in \mathcal{L}^1(\mathcal{F}_s)$ can be approximated by an almost surely converging sequence, and combined with the denseness result, the result extends to the multivariate case. 

\subsection{Simulation details}\label{appendix:sim-details}
In this section, we provide simulation details for the experiments in \cref{sec:simulation}.
\subsubsection{Details from \texorpdfstring{\cref{subsec:level-power}}{}}
Recall that from each structure, we sampled point processes with kernels $g^{i_1i_2}(s) = \alpha_{i_1i_2}\beta_{i_1i_2}e^{-\beta_{i_1i_2} s}$ if $(i_1, i_2) \in \mathcal{G}$ and otherwise $g^{i_1i_2}(s)=0$. We simulated data using the link-function $\eta(x) = 1_{x\geq 1}\cdot x + 1_{x<1}\cdot (\log(x)+1)$.

For all structures and edges, the decay parameter $\beta_{i_1i_2}$ is $0.8$, the baseline intensity is $\beta_0 = 0.25$ and the rate parameter on self-edges is $\alpha_{i_1i_1} = 0.4$.
The remaining rate parameters $\alpha_{i_1i_2}$ are given in \cref{fig:structures-details}. 

\subsubsection{Details from \texorpdfstring{\cref{subsec:sim-eca}}{}}\label{subsec:details-eca}
All graphs are sampled randomly with all self-edges present and all other edges sampled with a probability of an edge occuring at $p=0.2$. 
Given the graph, Hawkes processes are sampled with kernels $g^{i_1i_2}(s) = \alpha_{i_1i_2}\beta_{i_1i_2}e^{-\beta_{i_1i_2} s}$ if $(i_1, i_2) \in \mathcal{G}$ and otherwise $g^{i_1i_2}(s)=0$, and again using the link function $\eta(x) = 1_{x\geq 1}\cdot x + 1_{x<1}\cdot (\log(x)+1)$. 
The decay parameter is $\beta^{i_1i_2} = 0.8$, the baseline intensities $\beta_0^{i_1} = 0.25$, and for self-edges the rate parameter is $\alpha_{i_1i_1}=0.3$. The rate parameters between two different nodes is $s\cdot 0.4$ where $P(s=1) = P(s=-1) = 1/2$. 

\subsection{Estimated graphs for remaining 4 experiments}
\Cref{fig:neuron_firing} in \cref{sec:real-world} we displayed data and resulting estimated graphs from the first repetition in an experiment that was repeated $5$ times.  
This section contains plots similar to \cref{fig:neuron_firing}, but for the other $4$ repetitions. These are displayed in \cref{fig:rep-2,fig:rep-3,fig:rep-4,fig:rep-5}.
\begin{figure}[ht]
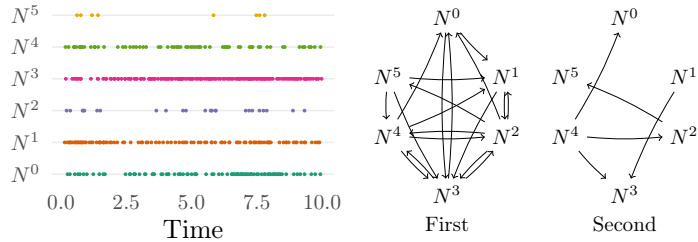

	\centering
    \subfloat{
    \resizebox{0.3\linewidth}{!}{%
}
    }
    \caption{Data and estimated graphs from repetition 2}
    \label{fig:rep-2}
\end{figure}
\begin{figure}[ht]
	\centering
    \subfloat{
    \resizebox{0.3\linewidth}{!}{%
%
}
    }
    \caption{Data and estimated graphs from repetition 3}
    \label{fig:rep-3}
\end{figure}
\begin{figure}[ht]
	\centering
    \subfloat{
    \resizebox{0.3\linewidth}{!}{%
%
}
    }
    \caption{Data and estimated graphs from repetition 4}
    \label{fig:rep-4}
\end{figure}
\begin{figure}[ht]
	\centering
    \subfloat{
    \resizebox{0.3\linewidth}{!}{%
%
}
    }
    \caption{Data and estimated graphs from repetition 5}
    \label{fig:rep-5}
\end{figure}

\end{document}